\documentclass[11pt,american,twoside,a4paper]{article}

\usepackage[T1]{fontenc}
\usepackage[utf8]{inputenc}
\usepackage{textalpha}
\usepackage[american,provide=*]{babel}
\usepackage{a4wide}
\usepackage{fourier,bbm,bm}
\usepackage{BOONDOX-calo}
\usepackage{amsmath,amsfonts,amsthm,amssymb}
\usepackage{color,xcolor,graphicx}
\usepackage{latexsym,mathtools,leftindex}
\usepackage[shortlabels]{enumitem}
\usepackage{xargs,xspace,needspace,xparse,easyReview}
\usepackage[active]{srcltx}
\usepackage[colorlinks=true]{hyperref}
\usepackage{fancyhdr}
\usepackage[toc]{appendix}
\usepackage[normalem]{ulem}
\usepackage[style=cap-alpha]{biblatex}

\colorlet{darkblue}{blue!50!black}
\hypersetup{
	colorlinks,%
	citecolor=darkblue,%
	filecolor=red,%
	linkcolor=darkblue,%
	urlcolor=darkblue,%
	pdfnewwindow=true,%
	pdfstartview={FitH}
}

\numberwithin{equation}{section}
\theoremstyle{plain}
\newtheorem{theorem}{Theorem}[section]
\AtBeginEnvironment{theorem}{\Needspace{5\baselineskip}}
\newcommand\bet{\begin{theorem}} \newcommand\eet{\end{theorem}}
\newtheorem{proposition}[theorem]{Proposition}
\AtBeginEnvironment{proposition}{\Needspace{5\baselineskip}}
\newcommand\bep{\begin{proposition}} \newcommand\eep{\end{proposition}}
\newtheorem{lemma}[theorem]{Lemma}
\AtBeginEnvironment{lemma}{\Needspace{5\baselineskip}}
\newcommand\bel{\begin{lemma}} \newcommand\eel{\end{lemma}}
\newtheorem{corollary}[theorem]{Corollary}
\AtBeginEnvironment{corollary}{\Needspace{5\baselineskip}}
\newcommand\bec{\begin{corollary}} \newcommand\eec{\end{corollary}}
\theoremstyle{definition}
\newtheorem{definition}[theorem]{Definition}
\AtBeginEnvironment{definition}{\Needspace{5\baselineskip}}
\newcommand\bed{\begin{definition}} \newcommand\eed{\end{definition}}
\newtheorem{remark}{Remark}[section]
\AtBeginEnvironment{remark}{\Needspace{5\baselineskip}}
\newcommand\ber{\begin{remark}} \newcommand\eer{\end{remark}}
\newtheorem{remarks}[remark]{Remarks}
\AtBeginEnvironment{remarks}{\Needspace{5\baselineskip}}
\newcommand\bers{\begin{remarks}} \newcommand\eers{\end{remarks}}
\AtBeginEnvironment{quote}{\Needspace{5\baselineskip}}

\newcommand\beq{\begin{equation}} \newcommand\eeq{\end{equation}}

\newcounter{smallarabics}
\newenvironment{arabicenumerate}
{\begin{list}{{\normalfont\textbf{(\arabic{smallarabics})}}}
		{\usecounter{smallarabics}\setlength{\itemindent}{0cm}
			\setlength{\leftmargin}{5ex}\setlength{\labelwidth}{4ex}
			\setlength{\topsep}{0.75\parsep}\setlength{\partopsep}{0ex}
			\setlength{\itemsep}{0ex}}}
	{\end{list}}
\newcommand{\ben}{\begin{arabicenumerate}} \newcommand{\een}{\end{arabicenumerate}}

\newcounter{smallroman}
\newenvironment{romanenumerate}
{\begin{list}{{\normalfont\textbf{(\roman{smallroman})}}}
		{\usecounter{smallroman}\setlength{\itemindent}{0cm}
			\setlength{\leftmargin}{5ex}\setlength{\labelwidth}{4ex}
			\setlength{\topsep}{0.75\parsep}\setlength{\partopsep}{0ex}
			\setlength{\itemsep}{0ex}}}
	{\end{list}}
\newcommand{\benr}{\begin{romanenumerate}} \newcommand{\eenr}{\end{romanenumerate}}

\newcommandx{\CA}[2][1=]{\todo[inline,author={ca},
	linecolor=green,backgroundcolor=green!15,bordercolor=green,#1]{#2}}
\newcommandx{\canot}[2][1=]{\todo[author={ca},
	linecolor=green,backgroundcolor=green!15,bordercolor=green,#1]{#2}}

\NewDocumentCommand{\ASS}{mm}{\expandafter\newcommand\csname #1\endcsname{{\hyperref[#1]{\bf (#2)}}}}
\NewDocumentCommand{\preASS}{mm}{\expandafter\newcommand\csname pre#1\endcsname{{\hyperref[#1]{\bf (#2)}}}}

\NewDocumentCommand{\CONJ}{mm}{\expandafter\newcommand\csname #1\endcsname{{\hyperref[#1]{\bf #2}}}}
\NewDocumentCommand{\preCONJ}{mm}{\expandafter\newcommand\csname pre#1\endcsname{{\hyperref[#1]{\bf #2}}}}
\newcommand{\conjecture}[3]{
	\begin{quote}\label{#2}{\bf#1\, }%
		#3%
	\end{quote}%
	\CONJ{#2}{#1}}

\newcommand\cA{{\mathcal A}} \newcommand\cB{{\mathcal B}} \newcommand\cC{{\mathcal C}}
  \newcommand\cF{{\mathcal F}}
 \newcommand\cH{{\mathcal H}} \newcommand\cI{{\mathcal I}}
 \newcommand\cK{{\mathcal K}} 
  
 \newcommand\cQ{{\mathcal Q}} \newcommand\cR{{\mathcal R}}
\newcommand\cS{{\mathcal S}} \newcommand\cT{{\mathcal T}}

 \newcommand\ce{{\mathcal e}}

\newcommand\cp{{\mathcal p}}  
\newcommand\cs{{\mathcal s}}

\renewcommand\AA{{\mathbbm A}}
\newcommand\NN{{\mathbbm N}}
\newcommand\ZZ{{\mathbbm Z}}

\newcommand\RR{{\mathbbm R}}
\newcommand\CC{{\mathbbm C}}

\newcommand{\one}{\mathbbm{1}}

\newcommand\fA{{\mathfrak A}} \newcommand\fB{{\mathfrak B}} \newcommand\fC{{\mathfrak C}}
  
  \newcommand\fI{{\mathfrak I}}

\newcommand\e{\mathrm{e}}

\renewcommand{\i}{\mathrm{i}}
\renewcommand\d{\mathrm{d}}

\newcommand{\Id}{\mathrm{Id}}

\newcommand{\argdot}{\,{\bm\cdot}\,}

\newcommand\myoperator[1]{\mathop{\mathrm{#1}}}

\newcommand\wstarlim{\myoperator{w^\ast-lim}\limits}
\newcommand\Dom{\myoperator{Dom}\nolimits}
\newcommand\Ker{\myoperator{Ker}\nolimits}
\newcommand\Ran{\myoperator{Ran}\nolimits}
\renewcommand\dim{\myoperator{dim}\nolimits}

\renewcommand\Im{\myoperator{Im}\nolimits}
\newcommand\sgn{\myoperator{sgn}\nolimits}

\newcommand\tr{\myoperator{tr}\nolimits}

\newcommand{\Aut}{\myoperator{Aut}\nolimits}

\newcommand{\dist}{\myoperator{dist}\nolimits}
\newcommand\supp{\myoperator{supp}\nolimits}

\newcommand\wstar[1]{$\mathrm{w}^\ast\!\!\mathrm{-#1}$}

\newcommand{\ds}{\displaystyle}
\renewcommand\proof{\noindent{\bf Proof.\ }}
\renewcommand\qed{\hfill$\Box$}
\newcommand\ie{\textsl{i.e.,\ }}
\renewcommand\bar{\overline}
\newcommand\ubar{\underline}
\newcommand{\twol}[2]{\genfrac{}{}{0pt}{1}{#1}{#2}}

\newcommand\eq{\mathrm{eq}}
\newcommand\req{\mathrm{req}}
\newcommand\re{\mathrm{sa}}
\newcommand\loc{\mathrm{loc}}

\newcommand{\diam}{\oslash}
\newcommand\Lim{\lim_\Lambda}
\newcommand\Llim{\lim_{\{\Lambda\}}}
\newcommand\fAloc{\fA_\loc}
\newcommand\fAr{\fA_\re}
\newcommand\cCphys{\mathcal{C}_\mathrm{phys}}
\newcommand\fCphys{\mathfrak{C}_\mathrm{phys}}
\newcommand\cBf{\cB_\mathrm{f}}
\newcommand\cBs{\cB_\mathrm{s}}
\newcommand\cBb{\cB_\mathrm{b}}
\newcommand\cBr{\cB^r}
\newcommand\cBsd{\cB_\mathrm{sd}}
\newcommand\cBR{\cB_\mathrm{R}}
\newcommand\cBdiam[1]{\cB^{\diam^{#1}}}
\newcommand\cSI{\cS_\mathrm{I}}
\newcommand\cSeq{\cS_\mathrm{eq}}
\newcommand\cSwg{\cS_\mathrm{wg}}
\newcommand\cSlwg{\cS_\mathrm{lwg}}
\newcommand\cSreq{\cS_\mathrm{req}}
\newcommand\cSKMS{\cS_\mathrm{KMS}}
\newcommand\cSphys{\cS_\mathrm{phys}}

\begin{document}
\renewcommand\today{}
\title{Approach to Equilibrium in Translation-Invariant\\
Quantum Systems: Structural Results II}
\author{V. Jak\v{s}i\'c$^{1}$, C.-A. Pillet$^{2,3}$, A. Szczepanek$^{4}$
	\\ \\ \\ 
	$^{1}$Dipartimento di Matematica, Politecnico di Milano\\
	piazza Leonardo da Vinci 32, 20133 Milano, Italy
	\\ \\
	$^2$Université de Toulon, CNRS, CPT, UMR 7332, 83957, La Garde, France
	\\ \\ 
	$^3$Aix-Marseille Univ, CNRS, CPT, UMR 7332, Case 907, 13288, Marseille, France
	\\ \\
	$^4$Institute of Mathematics, Jagiellonian University \\
	Łojasiewicza 6, 30-348 Kraków, Poland
}
\maketitle
\thispagestyle{empty}

\begin{flushright}
	\bigskip 
	
ἓν ηὐξήθη μόνον εἶναι ἐκ πλεόνων.

\smallskip
\emph{At one time they grew to be one alone out of many.}

\smallskip
Empedocles, \emph{On Nature}, fr.~B17.1--2
\end{flushright}

\bigskip\noindent{\small{\bf Abstract.}
We continue the development, started in
\href{https://doi.org/10.1007/s00023-023-01281-6}{[AHP~25~(2023), 715--749]}, of
a structural theory of approach to  equilibrium in translation-invariant quantum
spin systems. Dynamical Steady States (DSS) are defined as \wstar{limit\ points} 
of the time-averaged evolution of an initial state. We show that if a
DSS is an equilibrium state of a sufficiently regular interaction,
then the specific energy of this interaction is necessarily a constant of the
motion. Two distinct mechanisms lead to this conclusion, one based on dynamical
surface energies and the other on conservation of specific relative entropy. We
conjecture that these properties hold for a natural class of interactions.

For a DSS admitting a physical equilibrium description, the associated
equilibrium states are characterized by a dynamical Gibbs variational principle:
they are precisely the minimizers of the specific relative entropy from the
initial state over the physical equilibrium states. The achieved minimum
equals the increase of specific entropy.

The structural theory also yields characterizations of approach to thermal 
equilibrium, based on two model-dependent questions: restoration of a 
sufficiently local thermodynamic description at large times and 
classification of the corresponding local constants of motion. Once locality 
is restored, the constants of motion
determine the possible asymptotic equilibrium states. In particular, if
finite-range locality is restored and the energy is the only finite-range
constant of the motion, then every DSS is a thermal
equilibrium state at the temperature fixed by energy conservation. The
finite-time dynamical input underlying this theory is established in a
companion paper.
\par
}
\newpage

\tableofcontents

\newpage
\section{Introduction}

This paper is a direct continuation of~\cite{Jaksic2023a} and concerns the
development of a structural framework for approach to equilibrium in
translation-invariant lattice quantum spin systems. It is part of a series of
four papers devoted to this problem. The companion papers develop the
finite-time dynamical input~\cite{Jaksic2026b} to the structural theory
developed in this work, study the finite-range constants of motion in the XYZ
models~\cite{Jaksic2026c},  and finally illustrate the theory through the analysis of a class
of simple models~\cite{Jaksic2026e}.\footnote{At present, rigorous
large-time analysis is restricted to simple systems.} This series forms part of a broader research
program~\cite{Jaksic2026d}, briefly outlined at the end of this
Introduction.

The kinematical thermodynamic formalism of spin 
systems~\cite{Ruelle1978,Israel1979,Simon1993,Bratteli1981,Van_Enter1993} nee\-ded
for our study is reviewed in detail in Section~\ref{sec-QSS}, where we follow the
notation introduced in~\cite{Jaksic2023a}. We will not review this material in
the present Introduction. Instead, we start directly with a description of the
problem, a brief account of the results obtained in~\cite{Jaksic2023a}, and the
statement of the main results and conjectures of this work, introducing only the
notation needed for their formulation.

Let $\fA$ be the quantum spin $C^\ast$-algebra with configuration lattice
$\ZZ^d$ and finite-dimensional single spin Hilbert space $\cH$. $\cSI$ denotes
the set of all translation-invariant states on $\fA$ and $\cB_\mathrm{b/s}$ 
the usual big and small spaces of interactions. The specific energy and the
set of equilibrium states of an interaction $\Phi\in\cBb$ are 
denoted by $E(\Phi)$ and $\cSeq(\Phi)$. We denote by $\cBsd\subset\cBs$ the
class of all interactions $\Phi$ that generate
a $C^\ast$-dynamics, \ie a strongly continuous one-parameter group 
$(\tau_\Phi^\theta)_{\theta\in\RR}$ of $\ast$-automorphisms of $\fA$.
The work~\cite{Jaksic2023a} 
was focused on the space of exponentially decaying
interactions $\cB^r$, which is contained in $\cBsd$.  
The diameter spaces $\cBdiam{r}$, which are contained in $\cBsd$ for $r>d$, will also play an important role in this work.
Our starting point is a quantum dynamical system $(\fA,\tau_\Phi,\omega)$, with
$\Phi\in\cBsd$, the group $\tau_\Phi$ describing the time evolution 
of the system in the Heisenberg picture, and $\omega\in\cSI$
its state at time $t=0$.  The basic thermodynamic regularity assumption on this initial state is
that $\omega\in\cSeq(\Psi_0)$ for some interaction $\Psi_0\in\cBsd$.
To avoid discussion of trivialities, we will always assume that 
$\omega$ is not $\tau_\Phi$ invariant, and for $\nu\in\cSI$ we write
\beq
\nu_t=\nu\circ\tau_\Phi^t,\qquad
\bar\nu_T=\frac1T\int_0^T\nu_t\d t.
\label{eq:DSSDef}
\eeq
$\Sigma_+(\omega,\Phi)$ denotes the set of \wstar limit
points of the net $(\bar\omega_T)_{T>0}$. Since $\fA$ is separable,
$\omega_+\in\Sigma_+(\omega,\Phi)$ iff there exists a sequence $T_n\uparrow\infty$ such
that
\[
\omega_+(A)=\lim_{n\to\infty}\frac{1}{T_n}\int_0^{T_n}\omega\circ\tau_\Phi^t(A)\d t
\]
for all $A\in\fA$. Whenever the meaning is clear within the context, we will write
$\Sigma_+$ for $\Sigma_+(\omega,\Phi)$. The set $\Sigma_+$ is non-empty and its
elements, which are $\tau_\Phi$-invariant and translation-invariant, are called 
\emph{Dynamical Steady States}, abbreviated DSS.\footnote{They were called 
\emph{equilibrium steady states} (ESS) in~\cite{Jaksic2023a}.} 
In physically relevant models one expects that $\Sigma_+$ is a singleton.
\footnote{This is generally a model-dependent large-time question. In $d=1$, 
finite-range restoration of locality implies uniqueness of the DSS; see 
Theorem~\ref{thm:fr-rol-1d}.}

The first studies of approach to equilibrium in mechanical models go back to the
works of Maxwell and Boltzmann in the second half of the nineteenth century, who
recognized the central role this problem plays in foundations of thermodynamics.
The history of the problem from the perspective of our research program is
reviewed in the lecture notes~\cite{Jaksic2026d}. In a nutshell, our goal is to
understand under what conditions the map
\[
\omega\mapsto\Sigma_+
\]
corresponds to physical approach to equilibrium in the setting of quantum spin systems. 

As expected, symmetries and conservation laws play a fundamental role in this study. The
starting point is the finite-time conservation of the specific entropy and of
the specific energy of $\Phi$: for $\nu\in\cSI$ and $t\in\RR$,
$$
s(\nu_t)=s(\nu),\qquad\nu_t(E(\Phi))=\nu(E(\Phi)).
$$
They imply that for all $T>0$,\footnote{For the proof of the first relation 
see Proposition~2.16(1) in~\cite{Jaksic2023a}.} 
\[
s(\bar\nu_T)=s(\nu),\qquad 
\bar\nu_T(E(\Phi))=\nu(E(\Phi)).
\]
For $\Phi\in\cB^r$ the conservation of specific entropy was proven in 1968 by 
Lanford and Robinson~\cite{Lanford1968}. This result is extended 
to $\Phi\in\cBdiam{r}$ with $r>2d$ in~\cite{Jaksic2026b}. The conservation 
of specific energy was established for $\Phi\in\cBsd$ in~\cite{Jaksic2023a}.
\footnote{Although this result was stated there for $\Phi\in\cB^r$, the proof
extends directly to $\Phi\in\cBsd$.} It is an open question to characterize
the elements of $\cBsd$ for which the conservation of the specific entropy holds. 

For the rest of this introduction we assume that $\Phi\in\cBsd$ is such that 
conservation of the specific entropy holds. An immediate consequence of the 
upper-semicontinuity of the map $\cSI\ni\nu\mapsto s(\nu)$ is that for any 
$\omega_+\in \Sigma_+$, 
\[
s(\omega_+)\geq s(\omega).
\]
The validity of the strict inequality is known as the \emph{question of
Ruelle}~\cite{Ruelle1967}; see~\cite{Jaksic2023a} for a discussion. A related
observation concerns the interaction $\Psi_0$ associated to the initial state
$\omega$. Proposition~\ref{prop:energy-monotonicity}, which follows from the
Gibbs variational principle, gives, for any $\omega_+\in\Sigma_+$,
\[
\omega_+(E(\Psi_0))\geq \omega(E(\Psi_0)).
\]
The strict forms of the last two inequalities constitute the two
\emph{signatures of irreversibility},
\[
s(\omega_+)>s(\omega),\qquad
\omega_+(E(\Psi_0))>\omega(E(\Psi_0)).
\]
The work~\cite{Jaksic2023a} identified conditions under which these signatures
hold, with regularity and weak Gibbsianity playing an important structural
role. We will return to these points in the sequel.

The intuition that thermal equilibrium maximizes entropy subject to the conserved quantities 
goes back to Boltzmann's $H$-theorem and to Gibbs' variational principle. 
In our context this heuristic statement is formalized as follows. 

We say that a self-adjoint element $C\in\fA$ is a constant of motion for $\tau_\Phi$ if,
for all $\nu\in\cSI$ and $t\in\RR$,
\[
\nu_t(C)=\nu(C).
\]
$\fC(\Phi)$ denotes the set of all constants of motion for $\tau_\Phi$. For simplicity, 
in this Introduction we suppress the physical equivalence entering the precise definition; 
the latter is formulated in terms of equivalence classes in Section~\ref{sec:CM}.
$\fC(\Phi)$ is a closed real subspace of $\fA$ and $E(\Phi)\in\fC(\Phi)$.
The importance of $\fC(\Phi)$ in the study of approach to equilibrium stems 
from the obvious fact that for any $\omega_+\in\Sigma_+$ and $C\in\fC(\Phi)$, 
\[
\omega(C)=\omega_+(C).
\] 
Set 
\[
\cS_{\omega,\Phi}=\{\nu\in\cSI\mid\nu(C)=\omega(C)\text{ for all } C\in\fC(\Phi)\},
\]
and note that $\omega,\omega_+\in\cS_{\omega,\Phi}$. 

For our purposes, the following elementary observation is central.
\bep\label{prop-f-ent}
Suppose that $\omega_+\in\Sigma_+$ is such that, for some $\Psi_+\in\cBb$, one has 
$\omega_+\in\cSeq(\Psi_+)$ and $E(\Psi_+)\in\fC(\Phi)$. Then 
\[
s(\omega_+)=\max\{s(\nu)\mid\nu\in\cS_{\omega,\Phi}\}.
\]

\eep
\proof 
Let $\nu\in\cS_{\omega,\Phi}$. Since $E(\Psi_+)\in{\mathfrak C}(\Phi)$, the definition of
$\cS_{\omega,\Phi}$ and the fact that $\omega_+\in\Sigma_+$ give
\[
\nu(E(\Psi_+))=\omega(E(\Psi_+))=\omega_+(E(\Psi_+)).
\]
By the Gibbs variational principle, $s(\nu)-\nu(E(\Psi_+))\leq P(\Psi_+)$. 
On the other hand, since $\omega_+\in\cSeq(\Psi_+)$, $P(\Psi_+)=s(\omega_+)-\omega_+(E(\Psi_+))$. 
Combining the last three relations yields $s(\nu)\leq s(\omega_+)$. 
Since $\omega_+\in\cS_{\omega,\Phi}$, the maximum is attained at
$\omega_+$. \qed

This raises the question of when one can expect
\[
\omega_+\in\cSeq(\Psi_+),\qquad E(\Psi_+)\in{\mathfrak C}(\Phi),
\]
to hold for some physically reasonable interaction $\Psi_+$. 
The two parts of this assumption are of a rather different nature. 
\begin{itemize}
\item The first one is a large-time question. It requires the limiting DSS
$\omega_+$ to have a thermodynamic description. Such a circumstance cannot be
expected to follow from general structural considerations and has to be
established by a detailed analysis of concrete models. There are simple examples
where it fails.\footnote{Examples with finite-range classical interactions $\Phi$
are discussed in~\cite{Jaksic2026c,Jaksic2026d}.}

\item One of the main results of the present work is that, under suitable
structural assumptions, the second part follows from the first one. In other
words, once a DSS admits a sufficiently regular thermodynamic description, the
specific energy of the associated interaction is forced to be a constant of
motion for the driving dynamics $\tau_\Phi$.\footnote{The minimal regularity our
theory aims for is $\Psi_+\in\cBsd$. We will often return to this point in the
sequel.} Needless to say, if $\Sigma_+$ does not contain a DSS $\omega_+$
admitting a sufficiently regular thermodynamic description, then the dynamical
system $(\fA,\tau_\Phi,\omega)$ fails to describe approach to 
equilibrium.
\end{itemize}

The second requirement, that the specific energy of the interaction $\Psi_+$ is
a constant of motion for $\tau_\Phi$ can be rephrased in terms of Ruelle's
notion of physical equivalence;  see Section~\ref{sec:CM}. In particular,
additive constants are trivial under physical equivalence, and all statements
below concerning generators of spaces of constants of motion are understood in
this quotient sense.

Under our standing assumption of conservation of specific entropy, a
$\tau_\Phi$-invariant interaction in $\cBsd$ generates a dynamics commuting
with $\tau_\Phi$; see Theorem~\ref{thm:main2-const}.

Our first main result can now be formulated as follows.

\bet
Let $\omega_+\in\Sigma_+(\omega,\Phi)\cap\cSeq(\Psi_+)$ with
$\Psi_+\in\cBsd$. If the pair $(\Phi,\Psi_+)$ has the
\textbf{(SD)} property, then $E(\Psi_+)\in\fC(\Phi)$. 
\eet

The~\textbf{(SD)} property, defined precisely in Section~\ref{sec:CDC}, 
pertains to a pair $(\Phi,\Psi)$ of interactions. Roughly speaking, it ensures 
that, once evolved with the dynamics $\tau_\Phi^{-t}$ for small enough $t$,
the finite-volume Hamiltonians associated to $\Psi$ have well-defined surface 
energies which vanish per unit volume in the van~Hove limit, and, in this limit, 
generate the dynamics $\tau_\Phi^{-t}\circ\tau_\Psi\circ\tau_\Phi^t$. 
Thus~\textbf{(SD)} is a finite-time property not involving DSS or large-time limits.

We conjecture that any pair $(\Phi,\Psi)\in\cBsd\times\cBsd$ has the
\textbf{(SD)} property. The companion paper~\cite{Jaksic2026b} proves this
property when $(\Phi,\Psi)\in\cBdiam{r}\times\cBdiam{s}$ with $r>2d$
and $s>d$, and in particular for all finite-range pairs.

There is a second distinct route to the conclusion that
$E(\Psi_+)\in\fC(\Phi)$, based on specific relative entropy.

Recalling a notion introduced in~\cite{Jaksic2023a}, a pair
$(\eta,\Psi)\in\cSI\times\cBb$ is said to be \emph{regular} if, for every
$\nu\in\cSI$, the specific relative entropy $s(\nu|\eta)$ exists and satisfies
\beq
s(\nu|\eta)=-s(\nu)+\nu(E(\Psi))+P(\Psi).
\label{intro-regular-pair}
\eeq
This identity may be viewed as a relative-entropy counterpart of the Gibbs
variational principle. In particular, if $\eta\in\cSeq(\Psi)$, then
\[
s(\nu|\eta)=0\quad\Leftrightarrow\quad\nu\in\cSeq(\Psi).
\]
Recall our standing assumption that $\tau_\Phi$ preserves the specific entropy.
It then follows immediately from~\eqref{intro-regular-pair} that, for a regular
pair $(\eta,\Psi)$,
\[
s(\nu_t|\eta)=s(\nu|\eta)
\quad\hbox{for all $\nu\in\cSI$ and $t\in\RR$}
\]
if and only if $E(\Psi)\in\fC(\Phi)$. Under suitable assumptions,~\cite{Jaksic2026b} 
establishes conservation of the specific relative entropy along two state trajectories, 
\beq
s(\nu_t|\eta_t)=s(\nu|\eta).
\label{intro-re-conservation}
\eeq
If $\eta=\omega_+\in\Sigma_+$, then $\omega_+$ is $\tau_\Phi$-invariant 
and~\eqref{intro-re-conservation} becomes $s(\nu_t|\omega_+)=s(\nu|\omega_+)$. 
Consequently, whenever $(\omega_+,\Psi_+)$ is regular, the preceding
observation gives $E(\Psi_+)\in\fC(\Phi)$. In particular, when $\Phi\in\cBf$ 
this applies if $\|\Psi_+\|_r<r$ for some $r>0$, and also in dimension $d=1$ 
when $\Psi_+\in\cBf$; see~\cite{Jaksic2026b}.

The two mechanisms just described lead to two structural conjectures. The first
one is the \textbf{(SD)} conjecture stated above. The second one concerns
regularity and conservation of specific relative entropy. We conjecture that:
\textbf{(R)} for any $\Phi,\Psi\in\cBsd$ and $\eta\in\cSeq(\Psi)$, the pair
$(\eta,\Psi)$ is regular, and \textbf{(SE)}
\[
s(\nu\circ\tau_\Phi^t\,|\,\eta\circ\tau_\Phi^t)
=
s(\nu|\eta)
\]
for all $\nu\in\cSI$ and $t\in\RR$. We refer to these two statements as
Conjecture-\textbf{(R+SE)}. Conjectures-\textbf{(R)} and~\textbf{(SD)} together,
or Conjecture-\textbf{(R+SE)}, ensure that our theory applies whenever
$\omega_+\in\cSeq(\Psi_+)$ with $\Psi_+\in\cBsd$. In our opinion, these
conjectures are outstanding open problems in the kinematical and finite-time
theory of quantum spin systems whose formulation is \emph{a priori} unrelated to
the problem of approach to equilibrium studied here.

We now return to the variational interpretation of a DSS, and set
\[
\cBR=\left\{\Psi\in\cBb\,\middle|\,
(\nu,\Psi)\text{ is regular for every }\nu\in\cSeq(\Psi)\right\}.
\]
We further denote by $\cCphys(\Phi)$ the set of physical-equivalence
classes of interactions admitting a representative
$\Psi\in\cBsd\cap\cBR$ such that $E(\Psi)\in\fC(\Phi)$. 
Noting that $\cSeq(\Psi)$ only depends on the class $[\Psi]$ of $\Psi$,
we define \textsl{physical equilibrium states} as elements of
the set
\[
\cSphys(\Phi)=
\bigcup_{[\Psi]\in\cCphys(\Phi)}
\cSeq(\Psi); 
\]
a more precise formulation will be given in Section~\ref{sec:CM}. 

We shall refer to the next result as the \emph{dynamical Gibbs variational principle}.

\bet\label{thm:dGvp}
Let $\omega_+\in\Sigma_+(\omega,\Phi)$ be such that
$\omega_+\in\cSeq(\Psi_+)$ with $[\Psi_+]\in\cCphys(\Phi)$. Then, one has
\[
\inf_{\nu\in\cSphys(\Phi)}s(\omega|\nu)=s(\omega_+)-s(\omega),
\]
and the set of minimizers is precisely $\cSeq(\Psi_+)$.
\eet
This result complements Proposition~\ref{prop-f-ent} by providing a second
variational characterization of the equilibrium description of a DSS.
Proposition~\ref{prop-f-ent} is an
entropy-maximization statement on the constraint set determined by the constants
of motion. Theorem~\ref{thm:dGvp} is instead a relative-entropy minimization
statement on the set of physical equilibrium states: $\cSeq(\Psi_+)$ is
precisely the set of minimizers of $\nu\mapsto s(\omega|\nu)$ over
$\cSphys(\Phi)$. Whenever the quantum Stein lemma holds for the relevant pairs,
this characterization has the infor\-mation-theoretic interpretation that the
states in $\cSeq(\Psi_+)$ are the physical equilibrium states least
distinguishable from the initial state $\omega$ in asymptotic hypothesis
testing; see~\cite{Jaksic2012} and Theorem~\ref{thm:stein} below. There is also
a potential form of the dynamical Gibbs variational principle which does not
require regularity. It is technically useful, but does not have the direct
state-space interpretation above. We will discuss it in Section~\ref{sec:varcs}.

In dimension one and for finite-range interactions, the results are particularly
complete. For the finite-range class, \emph{finite-range restoration of
locality} means that every DSS $\omega_+\in\Sigma_+(\omega,\Phi)$ belongs to
$\cSeq(\Psi_+)$ for some $\Psi_+\in\cB_{\rm f}$.

\bet\label{thm-fin-r-mar}
Suppose that $d=1$, that $\Phi,\Psi_0\in\cB_{\rm f}$,
$\omega\in\cSeq(\Psi_0)$, and that $(\omega,\Phi)$ has the property of
finite-range restoration of locality. Then, the following hold:
\benr
\item $\Sigma_+(\omega,\Phi)=\{\omega_+\}$.
\item If $\omega_+\in\cSeq(\Psi_+)$ with $\Psi_+\in\cB_{\rm f}$, then
$[\Psi_+]\in\cCphys(\Phi)$.
\item The dynamical Gibbs variational principle holds and $\omega_+$ is its
unique minimizer.
\item $s(\omega_+)>s(\omega)$.
\eenr
\eet

At this point the structural picture can be summarized rather simply.
Suppose that
\[
\omega_+\in\Sigma_+(\omega,\Phi)\cap\cSeq(\Psi_+),
\qquad
\Psi_+\in\cBsd.
\]
If Conjectures-\textbf{(R)} and~\textbf{(SD)} hold, then $\Psi_+\in\cBR$ and
$E(\Psi_+)\in\fC(\Phi)$. Hence $[\Psi_+]\in\cCphys(\Phi)$ and the dynamical Gibbs
variational principle holds. The same conclusion follows directly from
Conjecture-\textbf{(R+SE)}. In~\cite{Jaksic2026b} the corresponding properties
are established in the regimes described above; in particular, they hold under
the assumptions of Theorem~\ref{thm-fin-r-mar}.

Thus, after separating off the structural properties of equilibrium states and
finite-time dynamics, one is left with a large-time question: does a DSS admit
an equilibrium description $\omega_+\in\cSeq(\Psi_+)$ with $\Psi_+$ belonging to
a sufficiently regular class of interactions? This part of the problem cannot be
decided by a general structural theory and has to be studied in concrete models.

We now turn to thermal equilibrium. Since
$E(\Phi)\in{\mathfrak C}(\Phi)$,
\[
e_\ast\coloneq\omega(E(\Phi))=\omega_+(E(\Phi))
\]
for every $\omega_+\in\Sigma_+$. We assume that $e_\ast$ belongs to the
interior of the range of the map
\[
\cSI\ni\nu\mapsto \nu(E(\Phi)).
\]
The Gibbs variational principle then determines the equilibrium inverse
temperature $\beta_\eq=\beta_\eq(\omega,\Phi)$ through the Legendre
duality relation
$$
\sup_{\substack{\nu\in\cSI\\ \nu(E(\Phi))=e_\ast}}s(\nu)
=
\inf_{\beta\in\RR}\left(P(\beta\Phi)+\beta e_\ast\right)
=
P(\beta_\eq\Phi)+\beta_\eq e_\ast,
$$
see Section~\ref{sec-sett-temp} for details.
We set
\[
\cS_\eq(\omega,\Phi)
=
\left\{
\nu\in\cSeq(\beta_\eq\Phi)\,\middle|\,
\nu(E(\Phi))=e_\ast
\right\}.
\]
Let $\omega_\eq\in\cS_\eq(\omega,\Phi)$. Since $\omega$ is not
$\tau_\Phi$-invariant, the Gibbs variational principle gives
$s(\omega)<s(\omega_\eq)$. Moreover, for every $\omega_+\in\Sigma_+$,
$$
s(\omega_+)\leq s(\omega_\eq),
$$
with equality if and only if $\omega_+\in\cS_\eq(\omega,\Phi)$. These 
are the kinematical entropic bounds associated with the equilibrium 
temperature $\beta_\eq$; see Proposition~\ref{prop:st-pres} below.

We next isolate a particularly simple mechanism leading to thermal
equilibrium. For a generic non-integrable system one expects that, in a
sufficiently local interaction class $\cB_0\subset\cBb$, the specific
energy is the only constant of the motion, \ie
$$
\{[\Psi]\mid\Psi\in\cB_0,E(\Psi)\in\fC(\Phi)\}=\operatorname{span}\{[\Phi]\}.
$$

\bet
Suppose that $\Phi\in\cBf$, that $(\omega,\Phi)$ has the property
of finite-range restoration of locality, and that
\[
\{[\Psi]\mid\Psi\in\cBf,E(\Psi)\in\fC(\Phi)\}
=\operatorname{span}\{[\Phi]\}.
\]
Then, one has
\[
\Sigma_+(\omega,\Phi)\subset\cS_\eq(\omega,\Phi).
\]
\eet
The result follows by combining finite-range restoration of locality with
the structural characterization of finite-range constants of motion; see
Theorem~\ref{thm:fr-ate} below. We emphasize that no regularity or weak-Gibbs
assumption is needed for this conclusion.

If additional constants of motion are present, thermalization requires the
initial and thermal states to agree on them. Suppose, for simplicity, that 
$\cS_\eq(\omega,\Phi)=\{\omega_\eq\}$. A necessary condition for a DSS to
coincide with $\omega_\eq$ is then
\beq
[\Psi]\in\cCphys(\Phi)\quad\Rightarrow\quad
\omega(E(\Psi))=\omega_\eq(E(\Psi)),
\label{intro-ATE-admissible}
\eeq
since every DSS has the same values as the initial state on all constants
of motion. We shall say that the pair $(\omega,\Phi)$ is
\emph{ATE-admissible} if~\eqref{intro-ATE-admissible} holds. The uniqueness
assumption on $\cS_\eq(\omega,\Phi)$ is made here only for simplicity; the
general criterion is given in Theorem~\ref{thm:main-ate} below.

\bet
Suppose that $\cS_\eq(\omega,\Phi)=\{\omega_\eq\}$, that
$\beta_\eq\Phi\in\cBR$, and that $(\omega,\Phi)$ is ATE-admissible.
Then, for any $\omega_+\in\Sigma_+(\omega,\Phi)$, the following statements
are equivalent:
\benr
\item $\omega_+=\omega_\eq$.
\item $\omega_+\in\cSphys(\Phi)$.
\eenr
\eet
Thus, in the presence of further conservation laws and without the
finite-range assumptions of the preceding result, the situation is more
subtle: a necessary admissibility condition involves the full family of physical
constants of motion. The dynamical Gibbs variational principle is designed
precisely to cover both situations.

The study of constants of motion in infinitely extended translation-invariant 
quantum spin systems is a separate research direction which, from the perspective
of our program, is initiated in the third paper of the series~\cite{Jaksic2026c}.
In this context, an important early contribution is Araki's analysis of constants 
of motion in the XY model~\cite{Araki1990}. We build on the recent works of Shiraishi, 
Chiba, Tasaki, and Sanatani~\cite{Shiraishi2019,Chiba2024,Chiba2025,Shiraishi2024b,
Yamaguchi2024a,Yamaguchi2024b,Shiraishi2025} who developed a detailed finite-volume 
theory of local conserved quantities for XYZ-type models.  
Building on these works, those of Hokkyo and Futami~\cite{HokkyoTest2025,FutamiHubbard2025},
and our dynamical characterization, the third paper develops a classification
theory of finite-range constants of motion in infinitely extended quantum
spin and fermion systems.

We emphasize that, in the case of thermal equilibrium, the structural theory
separates two distinct model-dependent problems. The first is to determine the 
sufficiently local constants of motion of a given dynamics. The second is 
asymptotic: determine whether the DSS admit equilibrium potentials in the
corresponding interaction class. 

We conclude this Introduction with several remarks.

The aim of our program is to separate the structural mechanisms of approach to
equilibrium -- conservation laws, locality, and variational principles -- from
the genuinely model-dependent features of the large-time dynamics. The program
is modeled on the kinematical theory of equilibrium quantum statistical
mechanics developed in the 1960s and 1970s by Araki, Ruelle, Lanford--Robinson,
Israel, and others. That theory isolated general structural principles -- the
Gibbs variational principle, the KMS characterization of equilibrium, physical
equivalence, etc. -- which provide the framework within which model-dependent
phenomena such as phase transitions can be studied. The theory developed here
may be viewed as a dynamical counterpart of this kinematical framework. This
analogy is also reflected in our terminology: dynamical steady states and the
dynamical Gibbs variational principle play roles parallel to equilibrium states
and the Gibbs variational principle. The dynamical ingredient which remains
outside the structural theory is restoration of locality. Once a large-time
state recovers a sufficiently local thermodynamic description, the structural
results developed here severely constrain, and in the regimes described above
determine, the possible dynamical steady states.

The four works in our series belong to a broader research program introduced
and discussed in detail in the forthcoming \textsl{Les Houches 2026} lecture
notes~\cite{Jaksic2026d}.\footnote{These lecture notes also contain extensions 
of the results of this work and its companions to translation-invariant 
fermionic systems.} This general program has three main directions. The first one, 
initiated by~\cite{Jaksic2023a} and the present series, concerns translation-invariant
lattice quantum spin systems. The second one deals with dynamical typicality and the 
interplay between large-time and thermodynamic limits. It is naturally  related with recent
rigorous works on approach to equilibrium based on the eigenstate thermalization 
hypothesis (ETH). The third direction is devoted to the study of approach to 
equilibrium in some classical systems, including anharmonic networks of classical 
oscillators and the classical Heisenberg model.

An announcement and pedagogical exposition of results of this paper and its
companions can be found in~\cite{Jaksic2026JMP}. 

The remainder of the paper is organized as follows. In Section~\ref{sec-QSS} we
review the kinematical thermodynamic formalism of quantum spin systems. 
Although much of this material is classical~\cite{Ruelle1978,Israel1979,Simon1993,
Bratteli1981,Van_Enter1993}, the length of the section is justified by the fact
that some of the notions introduced there appear to be new, and that some of the
results we need do not seem to be available in a convenient form in the literature.
Our main kinematical conjectures are also formulated in this section. 
Section~\ref{sec:ess} develops the theory of dynamical steady states, including
the equilibrium temperature and the basic entropic bounds. Section~\ref{sec:CM}
deals with constants of the motion and the two dynamical Gibbs variational principles. 
Section~\ref{sec:AE} formulates approach to equilibrium and restoration of locality, 
while Section~\ref{sec:ate-revisited} is devoted to approach to thermal equilibrium.

For convenience, several of the main results stated in the Introduction are
repeated below in the corresponding sections.

\paragraph*{Acknowledgments.} This work was partly funded by the CY
Initiative grant Investisse\-ments d’Avenir, grant number ANR-16-IDEX-0008. VJ
acknowledges the support of the MUR grant "Dipartimento di Eccellenza 2023-2027"
of Dipartimento di Matematica, Politecnico di Milano. We also acknowledge the
support of the ANR project DYNACQUS, grant number ANR-24-CE40-5714. 
AS is grateful to Clément Tauber for helpful discussions during the early stages of the project.

\section{Quantum Spin Systems}
\label{sec-QSS}

\subsection{Setting}
\label{sec-qss-setting}

\paragraph{$C^\ast$-algebra.} 

Let $\cH$ be the finite-dimensional Hilbert space of a single spin and consider
the lattice $\ZZ^d$ equipped with the norm $|x|=\sum_{j=1}^d|x_j|$. Let $\cF$ be
the collection of all finite subsets of $\ZZ^d$. We denote by $|X|$ the cardinality 
of $X\in\cF$ and, if $X\neq\emptyset$, by $\diam(X)\coloneq\max\{|x-y|\mid x,y\in X\}$ 
its diameter. For arbitrary $X\subset\ZZ^d$, we write $X^c\coloneq\ZZ^d\setminus X$.

A net $(\Lambda_\alpha)_{\alpha\in I}$ in $\cF$ is called a \emph{van~Hove net}
if it is exhaustive and, for every $r>0$,
\[
\lim_{\alpha\in I}\frac1{|\Lambda_\alpha|}
\left|\left\{x\in\Lambda_\alpha\mid
\dist(x,\Lambda_\alpha^c)\leq r\right\}\right|=0.
\]
We write $\Lim f(\Lambda)=L$ if $f(\Lambda_\alpha)\to L$ for every van~Hove net,
and $\Llim f(\Lambda)=L$ for the limit along the directed family of cubes
centered at the origin. The existence of a van~Hove limit implies the existence
of the centered-cube limit, with the same value, but not conversely. Unless
explicitly stated otherwise, thermodynamic limits in this paper are taken along
centered cubes. Van~Hove limits will be used below for surface-energy estimates
and in the \textbf{(SD)} property. The van~Hove formulation of the surface-energy 
condition is needed in Section~\ref{sec:CDC} for the application of the Araki--Moriya 
thermodynamic formalism~\cite{Araki2003}.

We shall also use inclusion-directed limits of a different nature. For fixed
$\Lambda\in\cF$, $\widetilde\Lambda$ may range over the finite supersets of
$\Lambda$, directed by inclusion. Such limits occur in the definition of surface
energies and are not thermodynamic limits.

Setting $\cH_x=\cH$ for all $x\in\ZZ^d$, we write $\cH_X\coloneq\bigotimes_{x\in X}\cH_x$,
$\fA_X\coloneq\cB(\cH_X)$ and let $\one_X$ denote the unit of $\fA_X$.\footnote{$\cB(\cH)$
denotes the $C^\ast$-algebra of all linear operators on the Hilbert space $\cH$.} 
If $Y\subset X$, we identify $\fA_Y$ with the subalgebra 
$\one_{X\setminus Y}\otimes\fA_Y\subset\fA_X$. For $X\subset\ZZ^d$, we denote by 
$\fA_X$ the inductive limit $C^\ast$-algebra over the family $(\fA_Y)_{\cF\ni Y\subset X}$.
The kinematical algebra of the spin system is $\fA_{\ZZ^d}$, and in the sequel
we will omit the subscript $\ZZ^d$. $\fAr$ will denote the real vector space of
all self-adjoint elements of $\fA$, and $\cS$ the set of states on $\fA$, the
positive and normalized elements of the dual $\fA^\ast$, equipped with the
\wstar{topology}. We shall further denote by $\cA$ the space of all
\wstar{continuous} affine maps $\ce:\cS\to\RR$, equipped with the sup-norm.

For any $X\subset\ZZ^d$ we have the identification $\fA=\fA_X\otimes\fA_{X^c}$.
An element $A\in\fA$ is called \textsl{local} if $A\in\fA_X$ for some $X\in\cF$.
The minimal $X$ for which this holds is called the \textsl{support} of $A$ and is
denoted by $\supp(A)$. The set $\fA_\loc$ of local elements is a dense
$\ast$-subalgebra of $\fA$.

For $X\in\cF$, let $\tr_{\cH_X}:\cB(\cH_X)\to\CC$ be the usual trace. If
$A\in\fA_X$, we further denote by
$$
\tr(A)\coloneq\frac{\tr_{\cH_X}(A)}{\tr_{\cH_X}(\one)}
$$
the normalized trace of $A$, which does not depend on the choice of 
$X\supset\supp(A)$. Since $|\tr(A)|\le\|A\|$, the normalized trace 
extends by continuity to a state $\tr\in\cS$, known as the unique 
tracial state on $\fA$, see~\cite[Section~6.2.5]{Bratteli1981}. 
Given an arbitrary subset $X\subset\ZZ^d$, we denote by 
$$
\begin{array}{rccc}
\tr_X:&\fA_X\otimes\fA_{X^c}&\to&\fA_{X^c}\\
&\sum_kA_k\otimes B_k&\mapsto&\sum_k\tr(A_k)B_k
\end{array}
$$ 
the partial normalized trace over $\fA_X$.

\paragraph{$C^\ast$-dynamical systems and KMS-states.} 

We denote by $\Aut(\fA)$ the group of $\ast$-automorphisms of $\fA$. A
$C^\ast$-\textsl{dynamics} on $\fA$ is a strongly continuous one-parameter
subgroup $\tau=(\tau^t)_{t\in\RR}\subset\Aut(\fA)$.\footnote{$t\mapsto\tau^t$ is
strongly continuous iff $\lim_{t\to0}\|\tau^t(A)-A\|=0$ for all $A\in\fA$.}
Given a state $\omega\in\cS$, the pair $(\tau,\omega)$ is a
$C^\ast$-\textsl{dynamical system} on $\fA$. The state $\omega$ is said to be
$\tau$-\textsl{invariant} whenever $\omega\circ\tau^t=\omega$ for all $t\in\RR$.
Given $\beta\in\RR^\ast$,\footnote{We use the familiar notation
$\AA^\ast=\AA\setminus\{0\}$ for a ring $\AA$.} set
\[
S_\beta\coloneq\{z\in\CC\mid 0<\sgn(\beta)\Im z<|\beta|\}.
\]
A state $\omega$ is called a $\beta$-\textsl{KMS} state for $\tau$ if, for
any $A,B\in\fA$, there exists a function $F_{A,B}$, analytic on $S_\beta$,
bounded and continuous on its closure, such that
\[
F_{A,B}(t)=\omega(A\tau^t(B)),\qquad
F_{A,B}(t+\i\beta)=\omega(\tau^t(B)A),
\]
for all $t\in\RR$. Any $\beta$-KMS state for $\tau$ is $\tau$-invariant.
For later reference, we note the following properties of the spin algebra $\fA$;
see~\cite[Lemma~III.3.3 and Theorem~III.3.4]{Israel1979}.
\bet\label{thm:israel}
\benr
\item $\fA$ is simple, i.e., has no non-trivial two-sided ideals.
\item Any KMS state on $\fA$ is faithful. 
\item If two $C^\ast$-dynamics $\tau_1$ and $\tau_2$ on $\fA$
share a $\beta$-KMS state, then $\tau_1=\tau_2$. 
\eenr
\eet

\paragraph{Translation invariant states.} 

The spin algebra $\fA$ is uniformly asymptotically Abelian w.r.t.\;the 
natural action $\varphi:\ZZ^d\ni x\mapsto\varphi^x\in\mathrm{Aut}(\fA)$
of the translation group of the lattice $\ZZ^d$, \textsl{i.e.,}
\[
\lim_{|x|\to\infty}\|[\varphi^x(A),B]\|=0
\]
for all $A,B\in\fA$. We denote by $\fA_I^\ast$ the space of translation invariant
elements of the dual $\fA^\ast$. The set $\cSI\coloneq\fA_I^\ast\cap\cS$ of all 
translation-invariant states on $\fA$ is a compact convex subset of
$\cS$ and a Choquet simplex.

\paragraph{Specific entropies.} 

The restriction $\nu_\Lambda$ of a state $\nu\in\cSI$ to $\fA_\Lambda$ is
represented by a density matrix in $\fA_\Lambda$ which we denote by the same
letter: For any $A\in\fA_\Lambda$,
$$
\nu(A)=\nu_\Lambda(A)=\tr_{\cH_\Lambda}(\nu_\Lambda A).
$$
The von~Neumann entropy of $\nu_\Lambda$ is defined by
$S(\nu_\Lambda)\coloneq-\tr_{\cH_\Lambda}(\nu_\Lambda\log\nu_\Lambda)$ and is
non-negative. The \textsl{specific entropy} of $\nu$ is given by the limit
\beq
s(\nu)\coloneq\Llim\frac{S(\nu_\Lambda)}{|\Lambda|},
\eeq
which exists and defines an affine, \wstar{upper-semicontinuous} map
$\cSI\ni\nu\mapsto s(\nu)\in[0,\log\dim\cH]$, \cite[Proposition~6.2.38]{Bratteli1981}.
 
The relative entropy of $\nu_\Lambda$ w.r.t.\;the restriction $\mu_\Lambda$ 
of another state $\mu\in\cSI$ is defined by
$$
S(\nu_\Lambda|\mu_\Lambda)\coloneq\tr_{\cH_\Lambda}(\nu_\Lambda(\log\nu_\Lambda-\log\mu_\Lambda)).
$$
It satisfies $S(\nu_\Lambda|\mu_\Lambda)\ge0$, with equality iff 
$\nu_\Lambda=\mu_\Lambda$. The \textsl{specific relative entropy} of $\nu$ 
w.r.t.\;$\mu$ is given by the limit
\[
s(\nu|\mu)\coloneq\Llim\frac{S(\nu_\Lambda|\mu_\Lambda)}{|\Lambda|},
\]
whenever it exists, in which case we have $s(\nu|\mu)\geq0$.

\paragraph{Stein's Lemma.} 

Let $\nu,\mu\in\cSI$ be faithful states and $0<\epsilon <1$. The respective
\textsl{Stein error exponents} are defined by
\begin{align*}
\bar B_\epsilon&\coloneq\inf_{(T_\Lambda)}\,\limsup_{\{\Lambda\}}\frac1{|\Lambda|}\log\mu_\Lambda(\one-T_\Lambda),\\[4pt]
\ubar B_\epsilon&\coloneq\inf_{(T_\Lambda)}\,\liminf_{\{\Lambda\}}\frac1{|\Lambda|}\log\mu_\Lambda(\one-T_\Lambda),\\[4pt]
B_\epsilon&\coloneq\inf_{(T_\Lambda)}\,\lim_{\{\Lambda\}}\frac1{|\Lambda|}\log\mu_\Lambda(\one-T_\Lambda),
\end{align*}
where the infimums are taken over all families $(T_\Lambda)$ of orthogonal projections 
$T_\Lambda\in\fA_\Lambda$ such that $\nu_\Lambda(T_\Lambda)\leq\epsilon$. 
In the case of $B_\epsilon$ the families $(T_\Lambda)$ are also such that 
\[
\Llim\frac1{|\Lambda|}\log\mu_\Lambda(\one-T_\Lambda)
\]
exists. 

The Stein error exponents measure the distinguishability of the pair $(\nu,\mu)$
in the sense of asymptotic hypothesis testing of the family
$(\nu_\Lambda,\mu_\Lambda)$, or in the sense of how the supports of
$\nu_\Lambda$ and $\mu_\Lambda$ separate as $\nu_\Lambda \rightarrow \nu$ and
$\mu_\Lambda \rightarrow \mu$.

We say that the \textsl{quantum Stein Lemma} holds for the pair $(\nu,\mu)$ iff 
$s(\nu|\mu)$ exists and 
\[
\bar B_\epsilon=\ubar B_\epsilon=B_\epsilon=-s(\nu|\mu)
\]
holds for all $0<\epsilon<1$. The following is a special case of~\cite[Theorem~3.6]{Jaksic2012}.
\bet\label{thm:stein}
Suppose that for some $\delta >0$ the expression
\beq
[0,1+\delta]\ni \alpha\mapsto
e(\alpha)\coloneq\Llim\tfrac1{|\Lambda|}\log\tr_{\cH_\Lambda}(\nu_\Lambda^\alpha\mu_\Lambda^{1-\alpha})
\label{eq:lim}
\eeq
defines a continuous function $e$, differentiable on $(0,1]$ and such that $(\partial^+e)(0)<e^\prime(1)$.
\footnote{We denote by $\partial^\mp$ the left/right derivative.}
Then the quantum Stein Lemma holds for the pair $(\nu,\mu)$. 
\eet
Although the assumptions of this theorem are not optimal
(see~\cite[Section~6.5]{Jaksic2012}), they are technically and conceptually natural. Note
that $e^\prime(1)=s(\nu|\mu)$. We will discuss in Section~\ref{sec:finite-range}
specific settings in which the conditions of Theorem~\ref{thm:stein} are
satisfied.

\paragraph{Factoring out translation invariance.} 

Some structural aspects of $\fA$ related to translation invariance will play an
important role in our work. In the classical setting, the corresponding ideas
can be traced back to Ruelle~\cite{Ruelle1978}; we follow~\cite[Section~2.4.3]{Van_Enter1993}. 
We are not aware of a convenient reference for the corresponding quotient-space
results in quantum spin systems, so we include the short proofs. These results 
will be used to identify constants of the motion with classes of interactions.

Let $\fI\subset\fAr$ be the closed linear span of the set
$$
\fI_0\coloneq\left\{A-\varphi^a(A)\,\left|\,A\in\fAr,a\in\ZZ^d\right.\right\},
$$
and denote the Cesàro mean on $\fA$ by
$$
m_\Lambda(A)\coloneq \frac{1}{|\Lambda|}\sum_{x\in\Lambda}\varphi^x(A).
$$
\begin{proposition}\label{prop:vanish_on_average}
For $A\in\fAr$, the following conditions are equivalent. 
\benr
\item $A\in\fI$.
\item $\ds\Llim\|m_\Lambda(A)\|=0$.
\item $\varphi(A)=0$ for all $\varphi\in\fA^\ast_I$.
\item $\omega(A)=0$ for all $\omega\in\cSI$.
\eenr
\end{proposition}
\proof We shall use the following notation and properties, $C$ denoting a subset of a Banach space $B$,
and $D$ a subset of its dual $B^\ast$. The annihilators of $C$ and $D$ are, respectively,
\begin{align*}
C^\perp&\coloneq\{f\in B^\ast\mid f(x)=0\text{ for all }x\in C\},\\[4pt]
\leftindex^\perp D&\coloneq\{x\in B\mid f(x)=0\text{ for all }f\in D\}.
\end{align*}
We note that $C^\perp$ is a \wstar{closed} subspace of $B^\ast$ and $\leftindex^\perp D$
a norm-closed subspace of $B$. We recall also that $\leftindex^\perp(C^\perp)$ is the
norm closure of the linear span of $C$ while $(\leftindex^\perp D)^\perp$ is the
\wstar{closure} of the linear span of $D$; see~\cite[Theorem~4.7]{Rudin1973}.

\medskip\noindent\textbf{(i)} $\Rightarrow $\textbf{(ii)} For any $A\in\fA$ 
and $a\in\ZZ^d$ one has
$$
\left\|m_\Lambda\left(A-\varphi^a(A)\right)\right\|
\le\frac1{|\Lambda|}\sum_{x\in\Lambda\Delta(\Lambda+a)}\|\varphi^x(A)\|
\le\frac{|\Lambda\Delta(\Lambda+a)|}{|\Lambda|}\|A\|=O(|\Lambda|^{-1/d}),
$$
and since $m_\Lambda$ is a linear map such that $\|m_\Lambda(A)\|\le\|A\|$, 
the result follows from an $\epsilon/3$ argument.

\medskip\noindent\textbf{(ii)} $\Rightarrow $\textbf{(iii)} Immediately follows 
from the fact that $\varphi(A)=\varphi(m_\Lambda(A))$. 

\medskip\noindent\textbf{(iii)} $\Rightarrow $\textbf{(i)} 
By assumption $A\in\leftindex^\perp\fA_I^\ast$. Observing that $\fI_0^\perp\subset\fA_I^\ast$, 
we conclude $\fI=\leftindex^\perp(\fI_0^\perp)\supset\leftindex^\perp\fA_I^\ast\ni A$.

\medskip\noindent\textbf{(iii)} $\Leftrightarrow $ \textbf{(iv)} Since
$\fA^\ast_I$ is the linear span of $\cSI$,
one has $\leftindex^\perp\cS_I=\leftindex^\perp\fA_I^\ast$.\qed

Identifying $\RR$ with $\RR\one\subset\fAr$, we observe that $\fI$, and $\fI+\RR$ 
are closed subspaces of $\fAr$. Hence, the respective quotient 
spaces are Banach spaces for the associated norms:
\begin{gather*}
\|A\|_{\fAr/\fI}\coloneq\inf_{B\in\fI}\|A-B\|,\\[4pt]
\|A\|_{\fAr/(\fI+\RR)}\coloneq\inf_{B\in\fI,b\in\RR} \|A- B-b\one\|.
\end{gather*}
These quotient norms can be expressed in terms of Cesàro means.
\begin{proposition}\label{prop:seminorm_interpretation}
For any $A\in\fAr$, one has
\begin{gather}
\Llim\left\|m_\Lambda(A)\right\|=\|A\|_{\fAr/\fI},\label{eq:overI}\\
\Llim\left\|m_\Lambda(A)\right\|_{\fAr/\RR}=\|A\|_{\fAr/(\fI+\RR)}\label{eq:overIplusR}.
\end{gather}
\end{proposition}
\proof
Let $A\in\fAr$. On the one hand,
\beq
m_\Lambda(A)=A+\frac1{|\Lambda|}\sum_{x\in\Lambda}\left(\varphi^x(A)-A\right)=A+A_\Lambda,
\label{eq:mLambdaDecomp}
\eeq
where $A_\Lambda\in\fI$. It follows that $\|m_\Lambda(A)\|=\|A+A_\Lambda\|\ge\|A\|_{\fAr/\fI}$, 
and hence
\beq
\liminf_{\{\Lambda\}}\|m_\Lambda(A)\|\ge\|A\|_{\fAr/\fI}.
\label{eq:overIlow}
\eeq
On the other hand, for any $\epsilon>0$ there exists $B_\epsilon\in\fI$ such that 
$\|A-B_\epsilon\|<\|A\|_{\fAr/\fI}+\epsilon$. We deduce that
$$
\|m_\Lambda(A)\|\le\|m_\Lambda(A-B_\epsilon)\|+\|m_\Lambda(B_\epsilon)\|
\le\|A-B_\epsilon\|+\|m_\Lambda(B_\epsilon)\|
\le\|A\|_{\fAr/\fI}+\epsilon+\|m_\Lambda(B_\epsilon)\|,
$$
and hence, by Proposition~\ref{prop:vanish_on_average},
$$
\limsup_{\{\Lambda\}}\|m_\Lambda(A)\|\le\|A\|_{\fAr/\fI}+\epsilon.
$$
Letting $\epsilon\downarrow0$, the last estimate combined with~\eqref{eq:overIlow}
gives~\eqref{eq:overI}.

To prove~\eqref{eq:overIplusR} we derive from~\eqref{eq:mLambdaDecomp}
$$
\|m_\Lambda(A)\|_{\fAr/\RR}=\inf_{b\in\RR}\|A+A_\Lambda-b\one\|
\ge\inf_{B\in\fI,b\in\RR}\|A-B-b\one\|=\|A\|_{\fAr/(\fI+\RR)},
$$
and hence
\beq
\liminf_{\{\Lambda\}}\|m_\Lambda(A)\|_{\fAr/\RR}\ge\|A\|_{\fAr/(\fI+\RR)}.
\label{eq:overIplusRlow}
\eeq
To obtain the opposite estimate, we note that for  any $\epsilon>0$ 
there exists $B_\epsilon\in\fI$ and $b_\epsilon\in\RR$ such that 
$\|A-B_\epsilon-b_\epsilon\one\|<\|A\|_{\fAr/(\fI+\RR)}+\epsilon$. It follows that
\begin{align*}
	\|m_\Lambda(A)\|_{\fAr/\RR}
	&=\inf_{b\in\RR}\|m_\Lambda(A-B_\epsilon-b_\epsilon\one)+m_\Lambda(B_\epsilon)+(b_\epsilon-b)\one\|\\[4pt]
	&\le\inf_{b\in\RR}\left(\|A-B_\epsilon-b_\epsilon\one\|+\|m_\Lambda(B_\epsilon)\|+|b_\epsilon-b|\right)\\[4pt]
	&\le\|A\|_{\fAr/(\fI+\RR)}+\epsilon+\|m_\Lambda(B_\epsilon)\|,
\end{align*}
and hence, invoking again Proposition~\ref{prop:vanish_on_average},
$$
\limsup_{\{\Lambda\}}\|m_\Lambda(A)\|_{\fAr/\RR}\le\|A\|_{\fAr/(\fI+\RR)}+\epsilon.
$$
Combining this estimate with~\eqref{eq:overIplusRlow} and letting 
$\epsilon\downarrow0$ yields~\eqref{eq:overIplusR}.
\qed

\subsection{Thermodynamic formalism}
\label{sec:TF}

An \textsl{interaction} is a family $\Phi=(\Phi(X))_{X\in\cF}$, where 
each $\Phi(X)$ is a self-adjoint element of $\fA_X$, with the
convention $\Phi(\emptyset)=0$. In the following, without further
mention, all interactions are assumed to be translation invariant,
{\sl i.e.,} such that 
\[
\varphi^x(\Phi(X))=\Phi(X+x)
\]
holds for all $X\in\cF$ and $x\in\ZZ^d$. The set of all interactions has a
natural real vector space structure. For each $X\in\cF$, the \textsl{local
Hamiltonian} associated to the interaction $\Phi$ is the self-adjoint element of
$\fA_X$ defined by
$$
H_X(\Phi)\coloneq\sum_{Y\subset X}\Phi(Y).
$$ 

Alternatively, one can start with a family $H=(H_X)_{X\in\cF}$, 
where each local Hamiltonian $H_X$ is a self-adjoint element of $\fA_X$
such that $H_\emptyset=0$, and $\varphi^x(H_X)=H_{X+x}$ holds for all 
$X\in\cF$ and $x\in\ZZ^d$. Such a family of Hamiltonians $H$ is 
associated to a unique interaction $\Phi=(\Phi(X))_{X\in\cF}$ given
by the inclusion-exclusion formula 
\beq
\Phi(X)\coloneq\sum_{Y\subset X}(-1)^{|X|-|Y|}H_Y.
\label{eq:incexc}
\eeq
This interaction--Hamiltonian duality\footnote{See also the remark after
Proposition~6.2.3 in~\cite{Bratteli1981}.} will play a role in
Section~\ref{sec:CDC}, where we will study constants of the motion. Since, as we
shall describe below, physically reasonable conditions lead to mathematically
natural Banach spaces of interactions, the traditional approach to quantum spin
systems takes interactions as a starting point. We note however that the 
Araki--Moriya thermodynamic formalism~\cite{Araki2003}
treats interactions and Hamiltonians on equal footing. For reasons of space we
will not review this approach here, although we will make use of its results.

Various spaces of interactions appear in the  study of quantum spin systems. 
The most basic one is the space $\cBf$ of finite range interactions. 
We recall that an interaction $\Phi$ is called \textsl{finite range}
whenever $\sup\{\diam(X)\mid\Phi(X)\neq0\}<\infty$. 
 
The other common spaces arise by fixing a translation invariant weight
$h:\cF\to[1,\infty)$, setting 
\beq
\|\Phi\|_h\coloneq\sum_{0\in X}\frac{h(X)}{|X|}\|\Phi(X)\|,
\label{eq:h-norm}
\eeq
and $\cB_h\coloneq\{\Phi\mid\|\Phi\|_h<\infty\}$. The pair $(\cB_h,\|\argdot\|_h)$ 
is a Banach space and $\cBf$ is a dense subspace of $\cB_h$.

\subsubsection{$\cBb$ -- the big space of interactions} 

This space corresponds to the constant weight $h(X)=1$ for all $X$. 
We denote the induced norm by $\|\argdot\|_\mathrm{b}$ and observe that
$\cB_h\subset\cBb$ with $\|\Phi\|_h\geq\|\Phi\|_\mathrm{b}$ for any
weight $h$.

$\cBb$ provides a natural background for the thermodynamic formalism of lattice
quantum spin systems. To any interaction $\Phi\in\cBb$ one can associate 
the ``\textsl{specific energy}'' observable
$$
E(\Phi)\coloneq\sum_{X\ni0}\frac{\Phi(X)}{|X|}\in\fAr,
$$
and the induced affine map $\ce_\Phi\in\cA$ given by
\beq 
\ce_\Phi(\omega)\coloneq\omega(E(\Phi)).
\label{def:tpe}
\eeq
Note that, for any $\Phi\in\cBb$,
$$
\|\mathcal{e}_\Phi\|=\sup_{\omega\in\cS}|\ce_\Phi(\omega)|
=\|E(\Phi)\|\le\|\Phi\|_\mathrm{b},
$$
and that for any $\Lambda\in\cF$,
$$
\left\|\frac1{|\Lambda|}H_\Lambda(\Phi)\right\|
\le\frac1{|\Lambda|}\sum_{X\subset\Lambda}\|\Phi(X)\|
\le\frac1{|\Lambda|}\sum_{x\in\Lambda}\sum_{X\ni x}\frac{\|\Phi(X)\|}{|X|}
\le\|\Phi\|_\mathrm{b}.
$$

\paragraph{The Ruelle map $R$.} 

Define the \textsl{center} $c(X)$ of $X\in\cF$ as its $\lfloor(|X|+1)/2\rfloor$-th
element\footnote{$\lfloor\argdot\rfloor$ stands for the integer part.} in lexicographic
order, and set
\beq
R(\Phi)\coloneq\sum_{c(X)=0}\Phi(X),
\label{eq:RDef}
\eeq
for $\Phi\in\cBb$. To the best of our knowledge, the following observations go back to 
Ruelle~\cite[Section~3.2]{Ruelle1978}. 
\bep\label{prop:ru-R}
\benr
\item $R:\cBb\to\fAr$ is a surjective linear contraction such that $R(\cBf)=\fAr\cap\fAloc$.
\item For all $\Phi\in\cBb$, $R(\Phi)\equiv E(\Phi)\,(\mathrm{mod}\,\fI)$.
\item The map $\hat R:\cBb/\Ker(R)\ni\Phi+\Ker(R)\mapsto R(\Phi)\in\fAr$ is an isometric isomorphism.
\eenr
\eep
\proof
\textbf{(i)} Setting
$$
T_{X,Y}=\begin{cases*}
	1&\text{if $Y$ is a translate of $X$; }\\[4pt]
	0&\text{otherwise},
\end{cases*}
$$
for $\Phi\in\cBb$, one has
\beq
\begin{split}
	\|R(\Phi)\|\le\sum_{c(X)=0}\|\Phi(X)\|
	&=\sum_{c(X)=0}\frac{1}{|X|}\sum_{x\in X}\|\Phi(X-x)\|
	=\sum_{c(X)=0}\sum_{Y\ni0}\frac{\|\Phi(Y)\|}{|Y|}T_{X,Y}\\[4pt]
	&=\sum_{Y\ni0}\frac{\|\Phi(Y)\|}{|Y|}\sum_{c(X)=0}T_{X,Y}
	=\|\Phi\|_\mathrm{b},
\end{split}
\label{eq:Rtrick}
\eeq
where, in the last equality, we used the fact that $Y$ has a unique translate centered 
at $0$. This shows that $R$ is a linear contraction from $\cBb$ to $\fAr$. Clearly 
$R(\cBf)\subset\fAr\cap\fAloc$. Moreover, let $A\in\fAr\cap\fAloc$ and choose a finite set
$Y\supset\supp(A)$ centered at $0$. Define the translation-invariant finite-range
interaction $\Phi_A$ by
\[
\Phi_A(Y+a)=\varphi^a(A),\qquad a\in\ZZ^d,
\]
and $\Phi_A(X)=0$ for all other $X\in\cF$. Since $Y$ has a unique translate
centered at $0$, $R(\Phi_A)=A$. Moreover, exactly $|Y|$ translates of $Y$ contain
the origin, and hence $\|\Phi_A\|_\mathrm{b}=\|A\|$. Thus,
$R(\cBf)=\fAr\cap\fAloc$ and for any $A\in\fAr\cap\fAloc$, one has 
\beq
\|A\|=\inf\{\|\Phi\|_\mathrm{b}\mid \Phi\in\cBb,R(\Phi)=A\}.
\label{eq:Rbound}
\eeq
Now, for any $A\in\fAr$, there is a sequence $(A_n)_{n\in\NN}$ in $\fAr\cap\fAloc$ 
such that $\sum_n\|A_n\|<\infty$, and $A=\sum_nA_n$. By the above construction,
there is a sequence $(\Phi_n)_{n\in\NN}$ in $\cBf$ such that 
$R(\Phi_n)=A_n$ and $\|\Phi_n\|_\mathrm{b}=\|A_n\|$ for all $n\in\NN$. It follows that 
$\Phi=\sum_n\Phi_n\in\cBb$ and $R(\Phi)=A$. 

\medskip\noindent\textbf{(ii)} Repeating the argument in~\eqref{eq:Rtrick}, we derive 
$\omega(R(\Phi))=\omega(E(\Phi))$ for all $\omega\in\cSI$, and the claim
follows from Proposition~\ref{prop:vanish_on_average}.

\medskip\noindent\textbf{(iii)} With $\cK=\Ker(R)$, the quotient $\cQ=\cBb/\cK$ equipped
with the quotient norm $\|\argdot\|_\cQ$ is a Banach space. We denote by $\pi$ the
canonical surjection $\cBb\to\cQ$. Since $\fAr=\Ran(R)$ is closed, the map
$\hat R:\pi(\Phi)\mapsto R(\Phi)$ is an isomorphism, {\sl i.e.,}
a bounded linear bijection from $\cQ$ onto $\fAr$, with a bounded inverse. 
For $\Phi\in\cBf$, \eqref{eq:Rbound} can be written 
as
$$
\|\hat R\circ\pi(\Phi)\|=\inf\{\|\Phi+\Psi\|_\mathrm{b}\mid\Psi\in\cK\}=\|\pi(\Phi)\|_\cQ,
$$
and since $\pi(\cB_f)$ is dense in $\cQ$, this identity extends by continuity to $\Phi\in\cBb$.\qed

\paragraph{Factoring out translation invariance in $\cBb$.}

For any closed subspace $\fB\subset\fAr$, the Ruelle map $R$ induces an
isometric isomorphism
\[
\cBb/R^{-1}(\fB)\rightarrow\fAr/\fB.
\]
Indeed, for any $\Phi\in\cBb$, Proposition~\ref{prop:ru-R}(iii) gives
$$
\|R(\Phi)\|_{\fAr/\fB}=\inf_{B\in\fB}\|R(\Phi)+B\|
=\inf_{\Psi\in R^{-1}(\fB)}\|R(\Phi+\Psi)\|
=\inf_{\Psi\in R^{-1}(\fB)}\inf_{\Upsilon\in\Ker(R)}\|\Phi+\Psi+\Upsilon\|_\mathrm{b},
$$
and since $\Ker(R)\subset R^{-1}(\fB)$ we further have\footnote{To simplify notation, and
whenever the meaning is clear from the context, we shall
denote by $\Phi$ the equivalence class of $\Phi\in\cBb$ in a quotient space $\cBb/\cB$.\label{foot:Qnot}}
$$
\|R(\Phi)\|_{\fAr/\fB}
=\inf_{\Psi\in R^{-1}(\fB)}\|\Phi+\Psi\|_\mathrm{b}
=\|\Phi\|_{\cBb/R^{-1}(\fB)}.
$$
Of particular interest are the two cases $\fB=\fI$ and $\fB=\fI+\RR$, and
we denote the reciprocal images by $\cI=R^{-1}(\fI)$ and $\cR=R^{-1}(\RR)$.
                                               
\bet\label{thm:RuIso}
The Ruelle map $R$ induces the isometric isomorphisms 
\beq
\cBb/\cI\to\fAr/\fI,
\label{eq:hatR}
\eeq
and
$$
\cBb/(\cI+\cR)\to\fAr/(\fI+\RR).
$$
\eet

\begin{remark}
Instead of defining Ruelle's map by Equ.~\eqref{eq:RDef}, we could have chosen
$\Phi\mapsto E(\Phi)$. Indeed, it follows from Proposition~\ref{prop:ru-R}(ii)
that the subspaces $\cI$ and $\cI+\cR$, as well as the isometries of Theorem~\ref{thm:RuIso},
are unaffected by this change. Other choices are possible, which are all
related to Ruelle's notion of physical equivalence that we discuss below.
\end{remark}

\paragraph{Specific energy, pressure, and variational principles}

The following summarizes the basic elements of the thermodynamic formalism.
The reader should consult~\cite{Ruelle1978,Israel1979} for a detailed exposition.

\bet\label{thm:Gibbs}
Assume that $\Phi\in\cBb$.
\benr
\item {\bf Specific energy.}
For any $\omega\in\cSI$, one has 
\[
\Llim\frac1{|\Lambda|}\omega(H_\Lambda(\Phi))=\ce_\Phi(\omega).
\]
\item {\bf Pressure.}
The limit 
\beq
P(\Phi)\coloneq\Llim\frac1{|\Lambda|}\log\tr_{\cH_\Lambda}\left(\e^{-H_\Lambda(\Phi)}\right)
\label{eq:PDef}
\eeq
exists and is finite. The map $\cBb\ni\Phi \mapsto P(\Phi)$ is convex,
continuous, and satisfies
\[
|P(\Phi)-P(\Psi)|\leq\|\Phi-\Psi\|_\mathrm{b}.
\]
\item {\bf Variational principle.} The pressure $P$ satisfies
\beq
P(\Phi)=\sup_{\omega\in\cSI}\left(s(\omega)-\ce_\Phi(\omega)\right).
\label{eq:GibbsPrince}
\eeq
The set of maximizers
\[
\cSeq(\Phi)\coloneq\left\{\omega\in\cSI\mid P(\Phi)=s(\omega)-\ce_\Phi(\omega)\right\}
\]
is a non-empty, convex, compact subset of $\cSI$. It is a Choquet simplex
and a face of $\cSI$. Its elements are called \textbf{equilibrium states}
for the interaction $\Phi$ and its extreme points are $\varphi$-ergodic.

\item {\bf Dual variational principle.} For any $\omega\in\cSI$, one has
\beq 
s(\omega)=\inf_{\Psi\in\cBb}\left(P(\Psi)+\ce_\Psi(\omega)\right). 
\label{eq:dual-var}
\eeq
\eenr
\eet

\begin{remark} The specific energy observable $E(\Phi)$ is not uniquely
determined by the relation
$$
\omega(E(\Phi))=\ce_\Phi(\omega)=\Llim\tfrac1{|\Lambda|}\omega(H_\Lambda(\Phi)).
$$
There are many $E\in\fAr$ for which the equality $\ce_\Phi(\omega)=\omega(E)$ 
holds for all $\omega\in\cSI$.
\footnote{For example, take $E=E(\Phi)+A$ where $A\in\fI$.} 
This basic observation provides one natural route to the notion
of physical equivalence that will play a central role in this paper.
\end{remark}
\begin{remark}\label{rem:fusina}
Much more is known about the thermodynamic formalism
and the general structure of $\cSeq(\Phi)$. Let us mention the following:
\ben 
\item For any weight $h$, the set $\{\Phi\in\cB_h\mid|\cSeq(\Phi)|=1\}$ 
is residual in $\cB_h$, \cite[Theorem~3(3)]{Lanford1968}.

\item On the other extreme, there is a dense subset of $\cBb$ whose elements
$\Phi$ are such that $\cSeq(\Phi)$ contains uncountably many $\varphi$-ergodic
states, \cite[Theorem~V.2.2(c)]{Israel1979}.

\item Although in general the infimum in~\eqref{eq:dual-var} may not be achieved, 
for any weight $h$ the set 
\[
\bigcup_{\Phi\in\cB_h}\cSeq(\Phi)
\]
is {\it norm} dense in $\cSI$, see Observation~3 in the proof 
of~\cite[Theorem~6.2.44]{Bratteli1981}. 

\item For any finite set $S=\{\omega_1,\ldots,\omega_N\}$ of $\varphi$-ergodic 
elements of $\cSI$ with finite specific entropy, there exists $\Phi\in\cBb$ such 
that $S$ is contained in the set of $\varphi$-ergodic elements of $\cSeq(\Phi)$, 
\cite[Theorem~V.2.2(a)]{Israel1979}.
\een
\end{remark}

\paragraph{Tangent functionals to the pressure.} 

As a corollary to Theorem~\ref{thm:RuIso}, the adjoint of the map~\eqref{eq:hatR} 
is an isometric isomorphism\footnote{Recall the notation introduced in the
proof of Proposition~\ref{prop:vanish_on_average}.}
\beq
\fI^\perp=\left(\fAr/\fI\right)^\ast\to\left(\cBb/\cI\right)^\ast=\cI^\perp.
\label{eq:RuAgain}
\eeq

Rewriting the variational formula~\eqref{eq:GibbsPrince} as
$$
P(\Phi)=\sup_{\omega\in\cSI}\left(s(\omega)-\omega(R(\Phi))\right),
$$ 
shows that the pressure map $P:\cBb\to\RR$ descends to the quotient $\cBb/\cI$.
Moreover, we observe that whenever $R(\Psi)=\one$, respectively $R(\Psi)\ge0$,
then $P(\Phi+c\Psi)=P(\Phi)-c$, respectively $P(\Phi+\Psi)\le P(\Phi)$.

In this perspective, $\alpha\in\cI^\perp$ is tangent to $P$ at $\Phi$ 
if $P(\Phi+\Psi)\ge P(\Phi)+\alpha(\Psi)$ for all $\Psi\in\cBb/\cI$.
By~\eqref{eq:RuAgain}, such a tangent functional takes the form 
$\alpha(\Psi)=-\mu(R(\Psi))$ for some $\mu\in\fI^\perp$. Taking 
the above observations into account, we further deduce that
$\mu\in\cSI\cap\fI^\perp$, and hence 
\beq
\alpha(\Psi)=-\ce_\Psi(\mu).
\label{eq:tangentForm}
\eeq 
With this in mind, the next result shows that tangent functionals to
the pressure at $\Phi$ are in one-to-one correspondence with equilibrium
states for $\Phi$, see~\cite{Ruelle1978,Israel1979}.

\bet\label{thm:tangent} 
$\omega\in\cSeq(\Phi)$ iff, for all $\Psi\in\cBb$, 
\[
P(\Phi+\Psi)-P(\Phi)\ge-\ce_\Psi(\omega).
\]
\eet

\paragraph{Ruelle's physical equivalence in $\cBb$.} 

\bep\label{prop:two-rpe} 
Let $\Phi,\Psi\in\cBb$ and $c\in\RR$. Then the following statements are equivalent.
\benr 
\item $\Phi-\Psi\in R^{-1}(\fI+c\one)$.
\item 
\[
\Llim\left(m_\Lambda(E(\Phi))-m_\Lambda(E(\Psi))\right)=c\one.
\]
\item
\[
\Llim\frac1{|\Lambda|}(H_\Lambda(\Phi)-H_\Lambda(\Psi))=c\one.
\]
\item For any $\omega\in\cSI$, 
$$
\ce_\Phi(\omega)-\ce_\Psi(\omega)=c.
$$
\item For any $\omega\in\cSI$, 
\[
\omega(E(\Phi))-\omega(E(\Psi))=c.
\]
\eenr
Moreover, any of these statements implies $P(\Psi)-P(\Phi)=c$ and $\cSeq(\Phi)=\cSeq(\Psi)$.
\eep

\proof
\textbf{(i)}$\Rightarrow$\textbf{(ii)} This is a direct consequence of 
Propositions~\ref{prop:vanish_on_average}(i)$\Rightarrow$(ii) and~\ref{prop:ru-R}(ii).

\medskip\noindent\textbf{(ii)}$\Rightarrow$\textbf{(iii)}
We first establish the estimate
\beq
\Llim\left\|\frac1{|\Lambda|}H_\Lambda(\Theta)-m_\Lambda(E(\Theta))\right\|=0
\label{eq:octopus}
\eeq
for every $\Theta\in\cBb$. Suppose first that $\Theta\in\cBf$ and that
$\Theta(X)=0$ for $\diam(X)>r$. For $\Lambda\in\cF$,
$$
\sum_{X\subset\Lambda}\Theta(X)-\sum_{x\in\Lambda}\varphi^x\left(\sum_{X\ni0}\frac{\Theta(X)}{|X|}\right)
=-\sum_{x\in\Lambda}\sum_{\twol{X\cap\Lambda^c\neq\emptyset}{X\ni x}}\frac{\Theta(X)}{|X|}.
$$
Hence, setting $\Lambda_r=\{x\in\Lambda\mid\dist(x,\Lambda^c)>r\}$,
\[
\left\|\frac1{|\Lambda|}H_\Lambda(\Theta)-m_\Lambda(E(\Theta))\right\|
\leq\frac{|\Lambda\setminus\Lambda_r|}{|\Lambda|}\|\Theta\|_\mathrm{b},
\]
and~\eqref{eq:octopus} follows along centered cubes. For general
$\Theta\in\cBb$, choose $\Theta_n\in\cBf$ with
$\|\Theta_n-\Theta\|_\mathrm{b}\to0$. Since, uniformly in $\Lambda$,
\[
\left\|m_\Lambda(E(\Upsilon))\right\|\leq\|\Upsilon\|_\mathrm{b},
\qquad
\frac1{|\Lambda|}\|H_\Lambda(\Upsilon)\|\leq\|\Upsilon\|_\mathrm{b},
\]
the finite-range estimate extends to $\Theta$.

Taking $\Theta=\Phi-\Psi$ in~\eqref{eq:octopus},
condition~\textbf{(ii)} yields~\textbf{(iii)}.

\medskip\noindent\textbf{(iii)}$\Rightarrow$\textbf{(iv)} is an immediate consequence of Theorem~\ref{thm:Gibbs}(i).

\medskip\noindent\textbf{(iv)}$\Rightarrow$\textbf{(v)} is a consequence of Definition~\eqref{def:tpe}.

\medskip\noindent\textbf{(v)}$\Rightarrow$\textbf{(i)} Invoking Propositions~\ref{prop:ru-R}~\textbf{(ii)}
and~\ref{prop:vanish_on_average}~\textbf{(i)} $\Rightarrow$~\textbf{(iv)} gives that $\omega(R(\Phi))-\omega(R(\Psi))=c$
for all $\omega\in\cSI$. The claim now follows from Proposition~\ref{prop:vanish_on_average}~\textbf{(iv)} $\Rightarrow$~\textbf{(i)}. 

\medskip\noindent The last statement follows from Part~(iv) and the Gibbs variational principle~\eqref{eq:GibbsPrince}.
\qed

For later reference, let us note the following corollary of this proof.
\bec\label{cor:Abel}
For all $\Phi\in\cBb$, one has $\Llim\left\|\frac1{|\Lambda|}H_\Lambda(\Phi)-m_\Lambda(E(\Phi))\right\|=0$.
\eec

The next definition is parallel to~\cite[Section~4.7]{Ruelle1978}.

\begin{definition} Two interactions $\Phi,\Psi\in\cBb$ are called
\textsl{physically equivalent,} written $\Phi\sim\Psi$, if
\[
R(\Phi-\Psi)\in\fI+\RR\one.
\]
They are called \textsl{strongly physically equivalent,} written
$\Phi\approx\Psi$, if
\[
R(\Phi-\Psi)\in\fI.
\]
These are equivalence relations compatible with the vector space structure of
$\cBb$.
\end{definition}

\begin{remark} Equivalently, $\Phi\sim\Psi$ iff
$\Phi\equiv\Psi\,(\mathrm{mod}\,\cI+\cR)$, and
$\Phi\approx\Psi$ iff $\Phi\equiv\Psi\,(\mathrm{mod}\,\cI)$.
By Theorem~\ref{thm:RuIso}, the same notation will be used for the induced
equivalence relations modulo $\fI+\RR\one$ and modulo $\fI$ on
$\fAr$. Proposition~\ref{prop:two-rpe} shows, in particular, that physical
equivalence is precisely equality of the specific energy functionals up to an
additive constant.
\end{remark}
\bep\label{prop:physeq} 
If $\Phi\sim\Psi$, then for any $t\in[0,1]$, 
\[
P((1-t)\Phi+t\Psi)=(1-t)P(\Phi)+tP(\Psi). 
\]
\eep

\proof By Proposition~\ref{prop:two-rpe}(iv), there is $c\in\RR$ such
that $\ce_\Psi(\omega)-\ce_\Phi(\omega)=c$ for all $\omega\in\cSI$.
The last assertion of Proposition~\ref{prop:two-rpe} gives
$P(\Psi)-P(\Phi)=-c$. Hence, by the Gibbs variational principle,
\[
\begin{split}
P((1-t)\Phi+t\Psi)
&=\sup_{\omega\in\cSI}
\big(s(\omega)-\ce_\Phi(\omega)-t(\ce_\Psi(\omega)-\ce_\Phi(\omega))\big)\\
&=P(\Phi)-tc=(1-t)P(\Phi)+tP(\Psi).
\end{split}
\]
\qed

The last result is complemented by 
\bep\label{prop:Pafin}
Suppose that $\Phi,\Psi\in\cBb$. Then the following statements are equivalent.
\benr
\item For any $t\in[0,1]$, 
\[
P((1-t)\Phi+t\Psi)=(1-t)P(\Phi)+tP(\Psi).
\]
\item $\cSeq(\Phi)\cap\cSeq(\Psi)\neq\emptyset$.
\eenr
\eep

\proof \textbf{(i)}$\Rightarrow$\textbf{(ii)} Choose
$\omega\in\cSeq((\Phi+\Psi)/2)$. By \textbf{(i)} and the
linearity of the specific energy in the interaction,
\[
s(\omega)=\frac12\big(\ce_\Phi(\omega)+P(\Phi)\big)
+\frac12\big(\ce_\Psi(\omega)+P(\Psi)\big).
\]
The variational principle gives
$s(\omega)\leq\ce_\Phi(\omega)+P(\Phi)$ and
$s(\omega)\leq\ce_\Psi(\omega)+P(\Psi)$.
Hence both inequalities are equalities, and
$\omega\in\cSeq(\Phi)\cap\cSeq(\Psi)$.

\medskip\noindent\textbf{(ii)}$\Rightarrow$\textbf{(i)} Let $\omega\in\cSeq(\Phi)\cap\cSeq(\Psi)$. 
Then the relations 
\[
s(\omega)=\ce_\Phi(\omega)+P(\Phi),\qquad s(\omega)=\ce_\Psi(\omega)+P(\Psi),
\]
give that, for $t\in[0,1]$, 
\[
s(\omega)=(1-t)s(\omega)+ts(\omega)=\ce_{(1-t)\Phi+t\Psi}(\omega)+(1-t)P(\Phi)+tP(\Psi).
\]
This equality and the dual variational principle, Theorem~\ref{thm:Gibbs}~\textbf{(iv)}, yield that 
\[
(1-t)P(\Phi)+tP(\Psi)\leq P((1-t)\Phi+t\Psi).
\]
The opposite inequality follows from the convexity of the pressure and~\textbf{(i)} follows. \qed

\paragraph{Weak Gibbsianity and regular pairs.} 

Set
\[
\omega_{\Phi,\Lambda}^\mathrm{c}\coloneq\e^{-H_\Lambda(\Phi)-P_\Lambda(\Phi)},\qquad 
P_\Lambda(\Phi)\coloneq\log\tr_{\cH_\Lambda}(\e^{-H_\Lambda(\Phi)}).
\]
\begin{definition}\label{def:wg}
A state $\omega\in\cSI$ is called \textsl{weak Gibbs} for the
interaction $\Phi\in\cBb$ if there exist constants $c_\Lambda>0$ satisfying
\begin{equation}
\Llim\frac{c_\Lambda}{|\Lambda|}=0,
\label{eq:mil-er}
\end{equation}
and such that, for all $\Lambda$,
$$
\e^{-c_\Lambda}\omega_{\Phi,\Lambda}^\mathrm{c}\leq\omega_\Lambda\leq \e^{c_\Lambda} \omega_{\Phi,\Lambda}^\mathrm{c}.
$$
$\omega$ is called \textsl{log weak Gibbs} for $\Phi$ if there exist
constants $c_\Lambda >0$ satisfying~\eqref{eq:mil-er} and such that, for all
$\Lambda$,
$$
- c_\Lambda+\log \omega_{\Phi,\Lambda}^\mathrm{c}\leq\log\omega_\Lambda
\leq c_\Lambda+\log\omega_{\Phi,\Lambda}^\mathrm{c}.
$$
We denote by $\cSwg(\Phi)$ the set of all weak Gibbs states for
$\Phi$, and by $\cSlwg(\Phi)$ the set of all log weak Gibbs states.
\end{definition}

The next results are from~\cite[Propositions~2.10,2.14]{Jaksic2023a}.\footnote{The proofs in~\cite{Jaksic2023a},
for $\Phi\in\cB^r$ (see Section~\ref{sec:exp-dec}), extend without change to $\Phi\in\cBb$.}

\bep\label{prop:weakGibbs_are_eq}
For $\Phi,\Psi\in\cBb$, the following hold:
\benr
\item $\omega\in\cSlwg(\Phi)\Leftrightarrow
\displaystyle\limsup_{\{\Lambda\}}\frac1{|\Lambda|}\|\log\omega_\Lambda-\log\omega_{\Phi,\Lambda}^c\|=0$.
\item $\cSwg(\Phi)\subset\cSlwg(\Phi)\subset\cSeq(\Phi)$.
\item Suppose that $\omega\in\cSlwg(\Phi)$. Then for any $\nu\in\cSI$, 
\beq
s(\nu|\omega)=-s(\nu)+\ce_\Phi(\nu)+P(\Phi).
\label{eq:def-rp}
\eeq
In particular $s(\nu|\omega)=0$ iff $\nu\in\cSeq(\Phi)$.
\eenr
\eep

\begin{definition}
A pair $(\omega,\Phi)\in\cSI\times\cBb$ is called regular if~\eqref{eq:def-rp} holds for any $\nu\in\cSI$.
\end{definition}

\preASS{WG}{WG}\preASS{LWG}{LWG}\preASS{R}{R}
\begin{definition}\label{WG}\label{LWG}\label{R}
\benr
\item A set $\cB\subset\cBb$ has the property \textbf{(WG)} if $\cSwg(\Phi)=\cSeq(\Phi)$ for any $\Phi\in\cB$.
\item A set $\cB\subset\cBb$ has the property \textbf{(LWG)} if $\cSlwg(\Phi)=\cSeq(\Phi)$ for any $\Phi\in\cB$.
\item A set $\cB\subset\cBb$ has the property \textbf{(R)} if, for any $\Phi\in\cB$ and any $\omega\in\cSeq(\Phi)$,
the pair $(\omega,\Phi)$ is regular. 
\eenr
\end{definition}

By Proposition~\ref{prop:weakGibbs_are_eq}~\textbf{(ii)} +~\textbf{(iii)}, \preWG$\Rightarrow$\preLWG$\Rightarrow$\preR.

\paragraph{Some comments on $\cBb$.} 
It is well-known that the space $\cBb$ is too big to properly describe physical 
systems, see~\cite{Bratteli1981, Israel1979}. Among the related problems, let us
mention: 

\ben 
\item The result of Israel~\cite{Israel1979} described in Remark~\ref{rem:fusina}(4)
is pathological from the physical point of view. In particular, it implies the existence
of physically non-equivalent interactions $\Phi,\Psi\in\cBb$ such that 
$\cSeq(\Phi)\cap\cSeq(\Psi)\not=\emptyset$. 

\item The surface energies
\beq 
W_\Lambda(\Phi)\coloneq\sum_{\twol{X\cap \Lambda\not=\emptyset}{X\cap \Lambda^c\not=\emptyset}}\Phi(X)
\label{eq:sur}
\eeq
are not defined for all $\Phi\in\cBb$. On physical grounds one expects that the surface
energies should play a central role in the study of equilibrium states.

\item The local Hamiltonians $H_\Lambda(\Phi)$ may not generate a $C^\ast$-dynamics on $\fA$.
Consequently, the dynamical aspects of quantum spin systems are out of reach.
\een

We proceed to describe various classes of interactions in $\cBb$ which are physically more natural 
for the study of quantum spin systems.

\subsubsection{$\cBs$ -- the small space of interactions.} 

This space corresponds to the choice $h(X)=|X|$ in~\eqref{eq:h-norm}. We denote the 
induced norm by $\|\argdot\|_\mathrm{s}$. The surface energies~\eqref{eq:sur} are 
well-defined for $\Phi\in\cBs$ and so is
$$
U_\Lambda(\Phi)\coloneq\sum_{X\cap\Lambda\neq\emptyset}\Phi(X)=H_\Lambda(\Phi)+W_\Lambda(\Phi).
$$
\bep\label{prop:BsStuf}
If $\Phi\in\cBs$, then the following hold. 
\benr
\item $\|W_\Lambda(\Phi)\|\le|\Lambda|\,\|\Phi\|_\mathrm{s}$.
\item $\ds\Lim\frac1{|\Lambda|}\|W_\Lambda(\Phi)\|=0$.
\item $\ds\Llim\left\|\frac1{|\Lambda|}U_\Lambda(\Phi)-m_\Lambda(E(\Phi))\right\|=0$.
\eenr
\eep

\proof \textbf{(i)} For $\Phi\in\cBs$, we have
\begin{align*}
\|W_\Lambda(\Phi)\|
&\leq\sum_{\twol{X\cap\Lambda\neq\emptyset}{X\cap\Lambda^c\neq\emptyset}}\|\Phi(X)\|
\leq\sum_{x\in\Lambda}\sum_{\twol{X\ni x}{X\cap\Lambda^c\neq\emptyset}}\|\Phi(X)\|\\
&=\sum_{x\in\Lambda}\sum_{\twol{X\ni0}{(X+x)\cap\Lambda^c\neq\emptyset}}\|\Phi(X)\|\\
&=\sum_{X\ni0}\|\Phi(X)\|
\left|\left\{x\in\Lambda\mid(X+x)\cap\Lambda^c\neq\emptyset\right\}\right|\\
&\leq|\Lambda|\sum_{X\ni0}\|\Phi(X)\|
=|\Lambda|\,\|\Phi\|_\mathrm{s}.
\end{align*}

\medskip\noindent\textbf{(ii)} Dividing the preceding estimate before its last
inequality by $|\Lambda|$ gives
\[
\frac1{|\Lambda|}\|W_\Lambda(\Phi)\|
\leq\sum_{X\ni0}\|\Phi(X)\|
\frac{\left|\left\{x\in\Lambda\mid(X+x)\cap\Lambda^c\neq\emptyset\right\}\right|}{|\Lambda|}.
\]
For fixed $X\ni0$, the set in the numerator is contained in
\[
\left\{x\in\Lambda\mid\dist(x,\Lambda^c)\leq\diam(X)\right\}.
\]
Its relative cardinality therefore tends to zero along every van~Hove net,
while it is bounded by one. Since
$\sum_{X\ni0}\|\Phi(X)\|=\|\Phi\|_\mathrm{s}<\infty$, dominated convergence
yields Part~\textbf{(ii)}.

\medskip\noindent\textbf{(iii)} Part~\textbf{(ii)} applies in particular along
centered cubes, and the assertion follows from Corollary~\ref{cor:Abel} and
$U_\Lambda(\Phi)=H_\Lambda(\Phi)+W_\Lambda(\Phi)$. \qed

To $\Phi\in\cBs$ one can associate the $\ast$-derivation $\delta_\Phi$ acting on $\fAloc$ 
by the formula 
\[
\delta_\Phi(A)\coloneq\sum_{X\in\cF}\i[\Phi(X),A].
\]
The following basic result holds; see~\cite[Theorem~3.2.50,Proposition~3.2.52]{Bratteli1987} 
and~\cite[Proposition~6.2.3]{Bratteli1981}.

\bet\label{thm:deltaPhi}
The derivation $\delta_\Phi$ is norm-closable on $\fAloc$.\footnote{We will always 
denote its closure by the same letter.} Its closure generates a $C^\ast$-dynamics 
$\tau_\Phi$ on $\fA$ if and only if $(\Id\pm\delta_\Phi)\fAloc$ is dense in $\fA$. 
In this case, for all $A\in\fA$, 
\[
\tau_\Phi^t(A)=\Llim\e^{\i tH_\Lambda(\Phi)}A \e^{-\i tH_\Lambda(\Phi)},
\]
and the limit is uniform for $t$ in compact subsets of\/ $\RR$. 
\eet

We denote by $\cBsd$ the set of all $\Phi$'s in $\cBs$ for which $(\Id\pm \delta_\Phi)\fAloc$ 
is dense in $\fA$. For $\Phi \in\cBsd$ we denote by $\cSKMS(\Phi)$ the set of \textit{translation invariant}
$1$-KMS states for $\tau_\Phi$.

\bet\label{thm:KMSeq} 
For $\Phi\in\cBsd$, 
\[ 
\cSKMS(\Phi)=\cSeq(\Phi).
\]
\eet
\begin{remark} The classical results in~\cite[Propositions~6.2.41 and~6.2.42]{Bratteli1981}
yield $\cSKMS(\Phi)\subset\cSeq(\Phi)$, with  equality 
established only in the case where $\Phi$ is in the space $\cBr$ defined in
Section~\ref{sec:exp-dec} below. The equality in the case $\Phi\in\cBsd$ 
is a consequence of~\cite[Theorems~A and~B]{Araki2003}.
\end{remark}

A first fact making $\cBsd$ a physically natural class of interactions is the following 
result which is a generalization of~\cite[Theorem III.4.2]{Israel1979}, where 
it is proven under the more restrictive assumption $\Phi,\Psi\in\cBr$.\footnote{The space $\cB^r\subset\cBsd$
is discussed in the next section.} 

\bet\label{thm:phys-eq} 
Suppose that $\Phi,\Psi\in\cBsd$. Then the following statements 
are equivalent: 
\benr
\item $\Phi\sim\Psi$. 
\item For any $t \in [0,1]$, 
\[ P((1-t)\Phi+t \Psi)=(1-t)P(\Phi) + t P(\Psi). \]
\item $\cSeq(\Phi)\cap\cSeq(\Psi)\neq\emptyset$.
\item $\tau_\Phi=\tau_\Psi$. 
\item $\cSeq(\Phi)=\cSeq(\Psi)$.
\eenr
\eet 

\proof 
\textbf{(i)}$\Rightarrow$\textbf{(ii)}$\Rightarrow$\textbf{(iii)} Follow from 
Propositions~\ref{prop:physeq} and~\ref{prop:Pafin} respectively.

\medskip\noindent\textbf{(iii)}$\Rightarrow$\textbf{(iv)} By Theorem~\ref{thm:KMSeq}, there is a state in $\cSI$
which is $1$-KMS for $\tau_\Phi$ and $\tau_\Psi$. Theorem~\ref{thm:israel}~\textbf{(iii)}
implies $\tau_\Phi=\tau_\Psi$.

\medskip\noindent\textbf{(iv)}$\Rightarrow$\textbf{(v)} Follows from Theorem~\ref{thm:KMSeq}.

\medskip\noindent\textbf{(v)}$\Rightarrow$\textbf{(i)} Clearly~\textbf{(v)}$\Rightarrow$~\textbf{(iii)} $\Rightarrow$~\textbf{(iv)}, 
so the claim reduces to~\textbf{(iv)} $\Rightarrow$~\textbf{(i)}. Since $\delta_\Phi=\delta_\Psi$, invoking Theorem~\ref{thm:deltaPhi} it suffices 
to show that any $\Phi\in\cBsd$ such that
\beq
\sum_{X\in\cF}[\Phi(X),A]=0
\label{eq:paperlog}
\eeq
for all $A\in\fAloc$ satisfies $\Phi\in\cI+\cR$. Now, for $A\in\fA_\Lambda$, we
can rewrite~\eqref{eq:paperlog} as $\left[U_\Lambda(\Phi),A\right]=0$. Taking the partial normalized trace over $\fA_{\Lambda^c}$ and using the
triviality of the center of $\fA_\Lambda$, we obtain
\[
H_\Lambda(\Phi)+\tr_{\Lambda^c}(W_\Lambda(\Phi))=c_\Lambda\one
\]
for some $c_\Lambda\in\RR$. Hence, by Proposition~\ref{prop:BsStuf}(ii)--(iii),
for any $\omega\in\cSI$,
\[
\omega(E(\Phi))=\Llim\frac{c_\Lambda}{|\Lambda|}\eqcolon c,
\]
and the claim follows from Proposition~\ref{prop:two-rpe}.\qed

A second reason for singling out $\cBsd$ as a physically natural class is
conservation of the specific energy, as stated in the following extension
of~\cite[Theorem~2.8(ii)]{Jaksic2023a}.
\bet\label{thm:energyp}
Suppose that $\Phi\in\cBsd$. Then
$$
\ce_\Phi(\omega\circ\tau_\Phi^t)=\ce_\Phi(\omega)
$$
holds for all $t\in\RR$ and $\omega\in\cSI$.
\eet

The proof, given in~\cite[Section~3.2]{Jaksic2023a} for $\Phi\in\cB^r$,
extends to $\Phi\in\cBsd$ by noticing that every occurrence of the norm $\|\argdot\|_r$
can be replaced by $\|\argdot\|_\mathrm{s}$.

We finish with a simple consequence of the Gibbs variational principle which
gives the non-strict form of the second signature of irreversibility.

\bep\label{prop:energy-monotonicity}
Suppose that $\Phi\in\cBsd$ is such that $\tau_\Phi$ preserves the specific
entropy, and let $\omega\in\cSeq(\Psi)$ for some $\Psi\in\cBb$. Then, for all
$t\in\RR$,
\[
\ce_\Psi(\omega\circ\tau_\Phi^t)\geq \ce_\Psi(\omega).
\]
Consequently, for every $\omega_+\in\Sigma_+(\omega,\Phi)$,
\[
\ce_\Psi(\omega_+)\geq \ce_\Psi(\omega).
\]
\eep

\proof
The Gibbs variational principle and conservation of the specific entropy give
\[
0\leq
-s(\omega\circ\tau_\Phi^t)
+\ce_\Psi(\omega\circ\tau_\Phi^t)+P(\Psi)
=
\ce_\Psi(\omega\circ\tau_\Phi^t)-\ce_\Psi(\omega),
\]
where in the last equality we used $\omega\in\cSeq(\Psi)$. Hence
\[
\ce_\Psi(\bar\omega_T)\geq\ce_\Psi(\omega)
\]
for all $T>0$. The last assertion follows by weak-$^\ast$ continuity of
$\ce_\Psi$. \qed

\subsubsection{$\cBr$ -- the spaces of exponentially decaying interactions} 
\label{sec:exp-dec}

These spaces correspond to the weights $h(X)=|X|\e^{r(|X|-1)}$ with $r>0$. 
The corresponding norms are 
\[
\|\Phi\|_r\coloneq\sum_{X\ni0}\e^{r(|X|-1)}\|\Phi(X)\|.
\]
As discussed in~\cite{Bratteli1981,Israel1979}, for many purposes $\cBr$
provides a natural framework for the quantum dynamics of lattice spin systems.
In particular, all results of~\cite{Jaksic2023a} 
were formulated in terms of this space of interactions.  

We summarize some of the basic properties of $\cBr$,
see~\cite[Theorems~2.3, 2.4, 2.8 and~2.11]{Jaksic2023a}. 
 
\bet\label{thm:first-s}
If $\Phi\in\cBr$, then: 
\benr
\item $\Phi\in\cBsd$ and Theorems~\ref{thm:KMSeq}, \ref{thm:phys-eq} and~\ref{thm:energyp} apply. 
\item For all $A\in\fAloc$ the map $\RR\ni t\mapsto\tau_\Phi^t(A)$
has an analytic extension to the strip $|\Im z|<r/2\|\Phi\|_r$, and for any 
$z$ in this strip, 
\[
\|\tau_\Phi^z(A)\|\leq\|A\|\e^{r|\supp(A)|}\frac{r}{r- 2\|\Phi\|_r|\Im z|}.
\]
\item If $r>\log\dim\cH$ and $\|\Phi\|_r<\frac1{2\dim\cH}$, then $\cSeq(\Phi)$ is a singleton, 
so that $\cSwg(\Phi)=\cSlwg(\Phi)=\cSeq(\Phi)$. 
\item If $\|\Phi\|_r<r$, then $\cSwg(\Phi)=\cSlwg(\Phi)=\cSeq(\Phi)$. 
\item $s(\omega)=s(\omega\circ\tau_\Phi^t)$ for all $\omega\in\cSI$ and $t\in\RR$. 
\eenr
\eet
 
In addition to Parts~\textbf{(iv)--(v)} of the above theorem, the following conservation law holds,
see~\cite{Jaksic2026b}. 
\bet\label{thm:pres-law}
Let $(\omega,\Psi)$ be a regular pair such that either
\benr
\item $\|\Psi\|_r<r$ for some $r>0$, or
\item $d=1$ and $\Psi\in\cBf$.
\eenr
Let $\Phi\in\cBf$ and $\nu\in\cSI$. Then there exists $T>0$ such that, for $|t|<T$, the specific relative
entropy of $\nu\circ\tau_\Phi^t$ with respect to
$\omega\circ\tau_\Phi^t$ exists and satisfies
$$
s(\nu\circ\tau_\Phi^t|\omega\circ\tau_\Phi^t)=s(\nu|\omega).
$$
One can take $T=\infty$ in case $\mathrm{(ii)}$ or if $\omega$ is $\tau_\Phi$-invariant.
\eet

\subsubsection{$\cBdiam{r}$-- the diameter spaces} 

In the previously introduced spaces of interactions, the weight $h$ was a
function of $|X|$ alone. Weights depending also on $\diam(X)$, the diameter of
$X$, naturally arise in the study of propagation properties of the
$C^\ast$-dynamics. A well-known example is
\[ 
h(X)=|X|^2(\dim\cH)^{2|X|}\e^{r\diam(X)},
\]
for $r>0$, see~\cite[Proposition 6.2.9]{Bratteli1981}. Building on more recent developments of
Lieb--Robinson bounds~\cite{Nachtergaele2019}, we introduce the weight
\[ 
h(X)= |X|^2(1+\diam(X))^r,
\]
and the associated norm $\|\argdot\|_{\diam^r}$, with $r>0$. The corresponding space is denoted  by 
$\cBdiam{r}$. Its basic properties are summarized in 
\bet\label{thm:second-s}
Suppose that $\Phi\in\cBdiam{r}$. 
\benr
\item If $r>d$, then $\Phi\in\cBsd$, and
Theorems~\ref{thm:KMSeq} and~\ref{thm:phys-eq} apply. 
\item If $d=1$ and $r\geq1$, then $\cSeq(\Phi)$ is a singleton. 
\item If $r>2d$, then
\[
s(\omega)=s(\omega\circ\tau_\Phi^t)
\]
for all $\omega\in\cSI$ and $t\in\RR$.
\eenr
\eet 

\proof
\textbf{(i)} This follows from~\cite[Theorem~4.8]{Bru2017}.

\medskip\noindent
\textbf{(ii)} If $d=1$ and $r\geq1$, then
\[
\sum_{X\ni0}\diam(X)\|\Phi(X)\|
\leq
\sum_{X\ni0}|X|(1+\diam(X))^r\|\Phi(X)\|
=
\|\Phi\|_{\diam^r}<\infty.
\]
The one-dimensional uniqueness theorem of Araki~\cite{Araki1975},
see also~\cite[Theorem~IV.6.1]{Simon1993}, therefore yields that
$\cSeq(\Phi)$ is a singleton.

\medskip\noindent
\textbf{(iii)} See the companion paper~\cite{Jaksic2026b}.
\qed

In this paper, we will work with the spaces $\cBsd$, $\cBdiam{r}$ for $r>d$, and $\cBr$. 
We emphasize that the last two spaces are \textsl{not comparable}: $\cap_{s>0}\cB^s$ is not contained 
in any $\cBdiam{r}$, and $\cap_{s>0}\cBdiam{s}$ is not contained in any $\cBr$.

\subsubsection{$\cBf$ -- finite range interactions} 
\label{sec:finite-range}
Recall that the space $\cBf$ of finite range interactions is dense
in $\cB_h$ for any $h$. In dimension $d=1$, these interactions have
special properties. The following result is implicit in
Araki's work~\cite{Araki1969}.

\bet\label{thm:pre-ara}
Suppose that $d=1$. For any $\Phi\in\cBf$, $\cSwg(\Phi)=\cSlwg(\Phi)=\cSeq(\Phi)$ is a singleton.
\eet

Using Araki's estimates, the companion paper~\cite{Jaksic2026b}
establishes the following one-dimensional case  of
Theorem~\ref{thm:pres-law}.
 
\bet\label{thm:pres-law-1}
Suppose that $d=1$ and $\Psi,\Phi\in\cBf$. If the pair $(\omega,\Psi)$ is regular
then, for all $\nu\in\cSI$ and $t\in\RR$, $s(\nu\circ\tau_\Phi^t|\omega\circ\tau_\Phi^t)$ 
exists and does not depend on $t$.
\eet

Going back to the Quantum Stein Lemma discussed in Section~\ref{sec-qss-setting}, the following
is known regarding Theorem~\ref{thm:stein}.

\bet\label{thm-qstein}
Suppose that $\Phi, \Psi\in \cBf$, and that $\nu\in\cSeq(\Phi)$, $\omega\in\cSeq(\Psi)$. 
\benr
\item If $d=1$, then the limit~\eqref{eq:lim} exists for all $\alpha\in\RR$, and the 
function $e$ is real analytic on $\RR$. Moreover, $e^\prime(0)<e^\prime(1)$ 
iff\, $\Phi\not\sim\Psi$. If $\Phi\sim\Psi$, then $e$ vanishes identically on $\RR$.
\item If $d>1$, then for any $a>1$ there is $\delta>0$ such that if 
$\|\Phi\|_\mathrm{s}<\delta$ and $\|\Psi\|_\mathrm{s}<\delta$, then 
the limit~\eqref{eq:lim} exists for all $\alpha\in(-a,a)$, and the function 
$e$ is real analytic on this interval. Moreover, $e^\prime(0)<e^\prime(1)$ iff $\Phi\not\sim\Psi$. 
If $\Phi\sim\Psi$, then $e$ vanishes identically on $(-a,a)$.
\eenr
\eet
Part~\textbf{(i)} was established in~\cite{Ogata2010} and Part~\textbf{(ii)}
in~\cite{Netocny2004}; see also~\cite[Section~3.4.2]{Jaksic2012}.
In the stated form, Theorem~\ref{thm-qstein} follows from these results
together with the weak Gibbsianity of $\nu$ and $\omega$ under the above
assumptions; see~\cite[Theorem~2.11]{Jaksic2023a}.
To the best of our knowledge, Theorem~\ref{thm-qstein} remains the
state-of-the-art result regarding the validity of the assumptions of
Theorem~\ref{thm:stein}.

\subsubsection{The WG, LWG and R conjectures}
\label{sec:R-conj}

\conjecture{Conjecture-(WG)}{conjWG}{$\cBsd$ has the property~\preWG.}
\conjecture{Conjecture-(LWG)}{conjLWG}{$\cBsd$ has the property~\preLWG.}
\conjecture{Conjecture-(R)}{conjR}{$\cBsd$ has the property~\preR.}

Recall that for a given interaction space $\cB\subset\cBb$,
\preWG$\Rightarrow$\preLWG$\Rightarrow$\preR. We formulate the three conjectures
separately because they have  different strengths and may well have
different optimal domains of validity. The intermediate LWG property is also
natural since the logarithm is operator monotone.

Only partial results are known. We regard these conjectures as outstanding
open problems in quantum statistical mechanics.

Weaker conjectures are made by replacing $\cBsd$ with any of the spaces $\cBr$, $\cBdiam{r}$ with $r>d$, or 
$\cBf$. They are also open except for $\cBf$ in $d=1$. In this case the powerful estimates of Araki 
\cite{Araki1969} establish \conjWG. The other known result towards resolution of these
conjectures is Theorem~\ref{thm:first-s}(iv). 

An alternative characterization of the \preWG{} property is given in~\cite[Proposition 2.12]{Jaksic2023a} and of the 
\preLWG{} property in~\cite[Proposition 2.13 (2)]{Jaksic2023a}. 
Although they are stated for interactions in $\cBr$, the proof of the characterization of the \preWG{} property holds for $\Phi \in \cBsd$ and of 
the \preLWG{} property for $\Phi\in\cBb$.

\section{Dynamical steady states}
\label{sec:ess}

\subsection{Definition and basic properties}

Throughout this section we work in the setting and notation of the Introduction.
In particular, $\tau_\Phi$ preserves specific entropy, while
Theorem~\ref{thm:energyp} gives conservation of specific energy. Hence
$$
s(\nu\circ\tau_\Phi^t)=s(\nu),\qquad
\ce_\Phi(\nu\circ\tau_\Phi^t)=\ce_\Phi(\nu)
$$
for all $t\in\RR$ and $\nu\in\cSI$. We also recall that these relations imply 
\beq
s(\bar\omega_T)=s(\omega),\qquad
\ce_\Phi(\bar\omega_T)=\ce_\Phi(\omega),
\label{eq:baromegaTconserv}
\eeq
for all $T>0$, and that for every
$\omega_+\in\Sigma_+(\omega,\Phi)$,
\beq
s(\omega_+)\ge s(\omega),\qquad
\ce_\Phi(\omega_+)=\ce_\Phi(\omega).
\label{eq:omegaplusprop}
\eeq

The following result strengthens the irreversibility result of~\cite{Jaksic2023a}.
\bet\label{thm:jpt-o}
Suppose that $\omega\in\cSeq(\Psi)$ for some $\Psi\in\cBb$. Then, for any
$\omega_+\in\Sigma_+(\omega,\Phi)$,
\benr
\item $\ce_\Psi(\omega_+)\ge\ce_\Psi(\omega)$.
\item If $\cSeq(\Psi)=\{\omega\}$, then $\ce_\Psi(\omega_+)>\ce_\Psi(\omega)$. 
\eenr
\eet
\begin{remark} In~\cite{Jaksic2023a} the strict statement was proven under the
assumption that $\Phi,\Psi\in\cBr$ and that $(\omega,\Psi)$ is regular. The proof
extends to $\Phi\in\cBsd$ and $\Psi\in\cBb$, assuming conservation of the specific
entropy. The proof below also dispenses with the regularity assumption.
\end{remark}
\proof
Let $(T_n)_{n\in\NN}$ be such that $\bar\omega_{T_n}\to\omega_+$. By the
Gibbs variational principle and~\eqref{eq:baromegaTconserv},
\begin{align*}
0
&\leq -s(\bar\omega_{T_n})+\ce_\Psi(\bar\omega_{T_n})+P(\Psi)\\
&=-s(\omega)+\ce_\Psi(\bar\omega_{T_n})+P(\Psi)\\
&=\ce_\Psi(\bar\omega_{T_n})-\ce_\Psi(\omega).
\end{align*}
Since $\ce_\Psi$ is \wstar continuous, letting $n\to\infty$ gives $\ce_\Psi(\omega_+)\geq\ce_\Psi(\omega)$,
which proves Part~\textbf{(i)}.

Suppose now that equality holds. Then
\[
-s(\bar\omega_{T_n})+\ce_\Psi(\bar\omega_{T_n})+P(\Psi)\rightarrow0.
\]
The map $\nu\mapsto -s(\nu)+\ce_\Psi(\nu)+P(\Psi)$
is \wstar lower semicontinuous. Hence
\[
0\leq -s(\omega_+)+\ce_\Psi(\omega_+)+P(\Psi)\leq0,
\]
and therefore $\omega_+\in\cSeq(\Psi)$. If
$\cSeq(\Psi)=\{\omega\}$, then $\omega_+=\omega$, contradicting the standing
assumption that $\omega$ is not $\tau_\Phi$-invariant. This proves
Part~\textbf{(ii)}. \qed

\subsection{Setting the equilibrium temperature}
\label{sec-sett-temp}

The conservation of the specific energy $\ce_\Phi$, together with the Gibbs
variational principle, determines the equilibrium temperature associated with
the pair $(\omega,\Phi)$. We begin with the thermodynamic input, formulated
as a parametric version of Theorem~\ref{thm:Gibbs}~\textbf{(iii)} +~\textbf{(iv)}.

Since $\omega$ is not $\tau_\Phi$-invariant, Theorem~\ref{thm:phys-eq}
implies that $\Phi\not\sim0$. Hence, by Theorem~\ref{thm:Gibbs}~\textbf{(ii)} and
Theorem~\ref{thm:phys-eq}~\textbf{(ii)} $\Leftrightarrow$~\textbf{(iv)}, the map
$\RR\ni\beta\mapsto\cp(\beta)\coloneq P(\beta\Phi)$ is continuous
and strictly convex. The specific energy $\ce_\Phi$ is a \wstar{continuous}
affine map which achieves its mini\-mum and maxi\-mum values $e_\mathrm{min}$ and $e_\mathrm{max}$
on the \wstar{compact} convex set $\cSI$. The specific entropy $s$ is
a \wstar{}up\-per-semicontinuous affine map on $\cSI$. Hence, for 
$e\in[e_\mathrm{min},e_\mathrm{max}]$, it achieves its maximum
value on the \wstar{compact} convex set $\{\nu\in\cSI\mid\ce_\Phi(\nu)=e\}$,
\beq
\cs(e)\coloneq\max\{s(\nu)\mid\nu\in\cSI,\ce_\Phi(\nu)=e\}\in[0,\log\dim\cH].
\label{eq:csDef}
\eeq
Setting $\cs(e)=-\infty$ for $e\in\RR\setminus[e_\mathrm{min},e_\mathrm{max}]$
defines an upper-semicontinuous concave function on $\RR$.

We can now rewrite the Gibbs variational principle~\eqref{eq:GibbsPrince} in the 
following parametric form
\beq
\cp(\beta)=\max_{\nu\in\cSI}\left(s(\nu)-\beta\ce_\Phi(\nu)\right)
=\max_{e\in\RR}(\cs(e)-\beta e).
\label{eq:AizLieb1}
\eeq
The last formula identifies $\cp$ with the Legendre--Fenchel transform 
of the closed proper convex function $\hat\cs:e\mapsto-\cs(-e)$.\footnote{This convex 
function is proper since it never takes the value $-\infty$ and is finite for 
$e\in[-e_\mathrm{max},-e_\mathrm{min}]$. It is lower-semicontinuous, and hence 
closed (see~\cite[Sections~4 and~7]{Rockafellar1970}).} Since $\cp$ is strictly convex,
$\hat\cs$ is essentially smooth (see~\cite[Section~26]{Rockafellar1970}).
This means that $\cs$ is differentiable on the open interval 
$(e_\mathrm{min},e_\mathrm{max})$, with a monotone decreasing derivative
satisfying 
$$
\lim_{e\downarrow e_\mathrm{min}}\cs'(e)=+\infty,\qquad
\lim_{e\uparrow e_\mathrm{max}}\cs'(e)=-\infty.	
$$
It will be convenient to denote this function by 
$\beta_\ast:[e_\mathrm{min},e_\mathrm{max}]\to[-\infty,+\infty]$.
\bet\label{thm:GibbsII}
\benr
\item There is a unique $e_0\in(e_\mathrm{min},e_\mathrm{max})$ such that 
$\cs(e_0)=\max\{\cs(e)\mid e\in\RR\}=\log\dim\cH$, the function $\cs$ being
monotone increasing on the interval $[e_\mathrm{min},e_0]$ and monotone decreasing
on the interval $[e_0,e_\mathrm{max}]$.
\item The parametric form of the Gibbs variational principle~\eqref{eq:AizLieb1}
admits the following dual principle: For any $e\in(e_\mathrm{min},e_\mathrm{max})$, 
\beq
\cs(e)=\min_{\beta\in\RR}\left(\cp(\beta)+\beta e\right),
\label{eq:smalldualGibbs}
\eeq
the minimum being achieved at $\beta=\beta_\ast(e)$, i.e.,
\beq
\cs(e)=\cp(\beta_\ast(e))+\beta_\ast(e)e.
\label{eq:AizLieb}
\eeq
Moreover, if $\nu\in\cSI$ is such that $\ce_\Phi(\nu)=e$, then $\cs(e)=s(\nu)$ iff
$\nu\in\cSeq(\beta_\ast(e)\Phi)$.
\eenr
\eet
\proof \textbf{(i)} The tracial state on $\fA$ is the unique maximizer of the
specific entropy $s$, \cite[Theorem~IV.2.4]{Simon1993}. Thus setting $e_0=\tr(E(\Phi))$,
we have $\cs(e_0)=\max\{s(\nu)\mid\nu\in\cSI\}=\log\dim\cH$. The function $\beta_\ast=\cs'$
being monotone decreasing, we have $\beta_\ast(e)\ge0$ for $e\in[e_\mathrm{min},e_0]$,
and $\beta_\ast(e)\le0$ for $e\in[e_0,e_\mathrm{max}]$.

\medskip\noindent\textbf{(ii)}
Since $\hat\cs$ is convex, it is the Legendre--Fenchel transform of $\cp$, 
which translates into Relation~\eqref{eq:smalldualGibbs}.
A well known property of the Legendre--Fenchel 
transform~\cite[Theorem~23.5]{Rockafellar1970},\footnote{$(\partial f)(x)$ 
denotes the subdifferential of $f$ at $x$, \ie the interval 
$[(\partial^-f)(x),(\partial^+f)(x)]$.}
$$
-e\in(\partial\cp)(\beta)
\quad\Leftrightarrow\quad \cp(\beta)-\cs(e)=-\beta e
\quad\Leftrightarrow\quad\beta=\beta_\ast(e),
$$
yields Relation~\eqref{eq:AizLieb}. Finally, it follows from definition~\eqref{eq:csDef}
that $s(\nu)\le\cs(e)$ for any $\nu\in\cSI$ such that $\ce_\Phi(\nu)=e$,
and Relation~\eqref{eq:AizLieb} yields that
$$
s(\nu)=\cs(e)=\cp(\beta_\ast(e))+\beta_\ast(e)e=P(\beta_\ast(e)\Phi)+\beta_\ast(e)\ce_\Phi(\nu)
$$
iff $\nu\in\cSeq(\beta_\ast(e)\Phi)$.
\qed

Going back to a pair $(\omega,\Phi)\in\cSI\times\cBsd$, let 
$e_\ast\coloneq\ce_\Phi(\omega)\in(e_\mathrm{min},e_\mathrm{max})$. Then 
$\beta_\eq=\beta_\eq(\omega,\Phi)\coloneq\beta_\ast(e_\ast)$ 
is our equilibrium inverse temperature, and
$$
\cs(e_\ast)=\inf_{\beta\in\RR}\left(\cp(\beta)+\beta e_\ast\right)=\cp(\beta_\eq)+\beta_\eq e_\ast.
$$
In the following, we set 
\[
\cS_\eq(\omega,\Phi)\coloneq\{\nu\in\cSeq(\beta_\eq(\omega,\Phi)\Phi)\mid\ce_\Phi(\nu)=\ce_\Phi(\omega)\}.
\]
\noindent\textbf{Remarks. (1)} A physical requirement is that $\beta_\eq\geq0$. This is tied to the
choice of the initial state and of the direction of time in the
Ces\`aro averages~\eqref{eq:DSSDef} defining the DSS. If
$\beta_\eq<0$, one may reverse the sign of $\Phi$, equivalently the
direction of time in~\eqref{eq:DSSDef}, so that the resulting equilibrium
inverse temperature is positive. See~\cite[Section~3.8]{Peierls1979}
for a discussion of this familiar issue.

\noindent\textbf{(2)} In~\cite{Jaksic2023a} we have set the equilibrium
temperature by rescaling $\Phi$, which allowed us to take $\beta_\eq=1$. This
equivalent approach was convenient for the purposes of~\cite{Jaksic2023a}.
However, in the present work it is more natural to fix $\Phi$ and set 
$\beta_\eq$ as above.

\noindent\textbf{(3)} If $e_\ast=e_0$, then $\beta_\eq=0$ and
$\cSeq(0\Phi)$ consists of the unique tracial state. Thus the thermal target
at $\beta_\eq=0$ is unambiguous. This case will reappear in the study of
long-range Ising-type models in the fourth paper in the series~\cite{Jaksic2026e} and in the lecture notes~\cite{Jaksic2026d}.

\noindent\textbf{(4)} If, in addition to the energy, one singles out finitely
many independent physical constants of the motion, the corresponding
thermodynamic parameters are obtained from the multivariable pressure by the
same Legendre--Fenchel procedure. The familiar chemical potential is the basic
example in fermionic systems. We shall not need this extension here, since the
role of the full family of constants of the motion will be handled directly by
the dynamical Gibbs variational principle.

 \subsection{Specific entropy bounds}

Once the equilibrium temperature is set, the following entropic bounds complement the 
first relation in~\eqref{eq:omegaplusprop}.
\bep\label{prop:st-pres}
Let $\omega_\eq\in\cSeq(\omega,\Phi)$. 
\benr 
\item $s(\omega)<s(\omega_\eq)$. 
\item 
For any $\omega_+ \in\Sigma_+(\omega, \Phi)$, $s(\omega_+)\leq s(\omega_\eq)$
with equality iff $\omega_+\in\cSeq(\omega,\Phi)$.
\eenr
\eep 
\proof \textbf{(i)} Since $\ce_\Phi(\omega)=\ce_\Phi(\omega_\eq)$, the dual variational 
principle~\eqref{eq:dual-var} gives
$$
s(\omega)=\inf_{\Psi\in\cBb}\left(P(\Psi)+\ce_\Psi(\omega)\right)\leq
P(\beta_\eq\Phi)+\beta_\eq\ce_{\Phi}(\omega_\eq)=s(\omega_\eq),
$$ 
with equality iff $\omega\in\cSeq(\beta_\eq\Phi)$. If $\beta_\eq\neq0$,
such a state is $\tau_\Phi$-invariant since it is invariant under
$\tau_{\beta_\eq\Phi}^t=\tau_\Phi^{\beta_\eq t}$. If $\beta_\eq=0$, it is the
tracial state and is again $\tau_\Phi$-invariant. Both cases contradict our
standing assumption.

\medskip\noindent\textbf{(ii)} For all $\nu\in\cSI$ such that 
$\ce_\Phi(\nu)=\ce_\Phi(\omega)\eqcolon e_\ast$ we have, by~\eqref{eq:smalldualGibbs}
and the variational definition of $\cs$,
$$
s(\omega_\eq)=P(\beta_\eq\Phi)+\beta_\eq\ce_{\Phi}(\omega)
=\cp(\beta_\eq)+\beta_\eq e_\ast=\cs(e_\ast)\ge s(\nu),
$$
with equality iff $\nu\in\cSeq(\omega,\Phi)$. By the second identity 
in~\eqref{eq:omegaplusprop}, this applies to $\nu=\omega_+$.
\qed

Energy conservation alone does not determine whether the upper entropic bound is
saturated. Additional constants of the motion may obstruct thermal equilibrium and,
more generally, constrain the thermodynamic description of a DSS. We therefore turn
to the structure of constants of the motion. The problem of approach to thermal
equilibrium is taken up again in Section~\ref{sec:ate-revisited} after the variational
theory has been developed.

\section{Constants of the motion}
\label{sec:CM}

\subsection{Conserved observables and invariant interactions}

For any $\Phi\in\cBsd$ and $A\in\fAr$, we set 
\[
\cT_\Phi(A)\coloneq\left\{t\in\RR\mid\nu\circ\tau_\Phi^t(A)=\nu(A)\text{ for all }\nu\in\cSI\right\}.
\]
By Proposition~\ref{prop:vanish_on_average}, $\cT_\Phi(A)$ is the set of all
$t\in\RR$ for which $\tau_\Phi^t(A)\approx A$. It is a closed subgroup of
$(\RR,+)$. In particular, if it contains an open neighborhood of $0$, then it
is equal to $\RR$.

\begin{definition}
An observable $A\in\fAr$ is called a \textsl{constant of the motion} of $\tau_\Phi$ if
\[
\nu\circ\tau_\Phi^t(A)=\nu(A)
\]
for all $\nu\in\cSI$ and $t\in\RR$, equivalently if $\cT_\Phi(A)=\RR$.
We denote by
\[
\fC(\Phi)
\coloneq
\left\{[A]\coloneq A+\fI+\RR\one\mid \cT_\Phi(A)=\RR\right\}
\subset \fAr/(\fI+\RR\one)
\]
the corresponding space of constants of the motion modulo physical
equivalence.
\end{definition}

A well known fact of mechanics, Noether's theorem, establishes a deep connection
between conservation laws and one-parameter groups of symmetries. Interactions, 
as generators of such symmetries, are therefore connected to constants 
of the motion. This motivates the next definition. 

\begin{definition}
An interaction $\Psi\in\cBb$ is said to be $\tau_\Phi$-\textsl{invariant} if its 
specific energy observable $E(\Psi)$ is a constant of the motion. We denote by
\[
\cC(\Phi)\coloneq
\left\{[\Psi]\coloneq\Psi+\cI+\cR\mid \Psi\in\cBb,[E(\Psi)]\in\fC(\Phi)\right\}
\subset\cBb/(\cI+\cR)
\]
the quotient space of physical equivalence classes of $\tau_\Phi$-invariant interactions.
\end{definition}

The machinery developed in Section~\ref{sec:TF} identifies the observable and
interaction descriptions of conservation laws.

\bep\label{prop:CMlevels}
The map
\[
\cC(\Phi)\ni[\Psi]\longmapsto [E(\Psi)]
\in\fC(\Phi)
\]
is an isometric isomorphism.
\eep
\proof By Proposition~\ref{prop:ru-R}(ii), $E(\Psi)\equiv R(\Psi)$ (mod $\fI$). 
The assertion is therefore the restriction of the second isometric
isomorphism of Theorem~\ref{thm:RuIso} to the subspaces singled out by the
conservation law.\vskip-8pt\qed

We shall see, in the next section, that under appropriate regularity conditions,
a $\tau_\Phi$-invariant interaction $\Psi$ indeed generates a symmetry of the
system, \ie a $C^\ast$-dynamics $\tau_\Psi$ commuting with $\tau_\Phi$.

From now on we identify $\cC(\Phi)$ and $\fC(\Phi)$ through Proposition~\ref{prop:CMlevels} 
and denote the resulting space simply by $\fC(\Phi)$. We will freely write either
$[A]\in\fC(\Phi)$ or $[\Psi]\in\fC(\Phi)$ according to whether an observable or an
interaction representative is being used.

Recall that 
\[
\cBR
\coloneq
\left\{\Psi\in\cBb\ \middle|\
(\nu,\Psi)\text{ is regular for every }\nu\in\cSeq(\Psi)\right\}.
\]
The \textsl{physical constants of the motion} are those elements of
$\fC(\Phi)$ which admit a representative in the thermodynamically and
dynamically natural class $\cBsd\cap\cBR$; explicitly,
\beq
\fCphys(\Phi)\coloneq
[\cBsd\cap\cBR]\cap\fC(\Phi).
\label{eq:Cphys-def}
\eeq

\subsection{Dynamical characterization}
\label{sec:CDC}

In this section we investigate the dynamical properties of $\tau_\Phi$-invariant
interactions. For $\Phi\in\cBsd$ and $X\in\cF$, we set 
$$
\tau_{\Phi,X}^t(A)=\e^{\i tH_X(\Phi)}A\e^{-\i tH_X(\Phi)}.
$$
Let $\Phi,\Psi\in\cBsd$. For any $t\in\RR$ we define the Hamiltonian 
$H(\Psi|\Phi,t)=\left(H_X(\Psi|\Phi,t)\right)_{X\in\cF}$ as 
\beq
H_X(\Psi|\Phi,t)=\tau_{\Phi,X}^{-t}(H_X(\Psi)).
\label{eq:ha-t}
\eeq
The van~Hove formulation in the next definition is used in the proof
of Theorem~\ref{thm:main2-const} under \textbf{(SD+)}, through the
application of~\cite[Proposition~14.1]{Araki2003} to the
transported Hamiltonians.

\preASS{SD}{SD}
\begin{definition}\label{SD}
We say that the pair $(\Phi,\Psi)\in\cBsd\times\cBsd$ has the \textbf{(SD)}\footnote{SD stands for surface-dynamics.}
property if the following hold for all $t$ in a neighborhood of $0$.
\benr
\item For every finite $\Lambda\in\cF$, the surface energy
\[
W_\Lambda(\Psi|\Phi,t)
=
\lim_{\widetilde\Lambda}
\big(H_{\widetilde\Lambda}(\Psi|\Phi,t)-H_\Lambda(\Psi|\Phi,t)
-H_{\widetilde\Lambda\setminus\Lambda}(\Psi|\Phi,t)\big)
\]
exists in norm, where $\widetilde\Lambda$ runs over the finite supersets of
$\Lambda$, directed by inclusion.
\item The surface energies vanish in the van Hove sense,
\[
\Lim\frac{1}{|\Lambda|}\|W_\Lambda(\Psi|\Phi,t)\|=0.
\]
\item One has, again in the van Hove sense,
\[
\Lim\frac1{|\Lambda|}\left\|\tau_\Phi^{-t}(H_\Lambda(\Psi))-H_\Lambda(\Psi|\Phi,t)\right\|=0.
\]
\eenr
\end{definition}
 
\bet Suppose that $(\Phi, \Psi)$ has the \preSD{} property and that for all $s,t\in\RR$, 
\beq
\tau_\Phi^t \circ \tau_\Psi^s=\tau_\Psi^s\circ \tau_\Phi^t.
\label{eq:DynCommut}
\eeq
Then $\Psi$ is $\tau_\Phi$-invariant, i.e., $[\Psi]\in\fC(\Phi)$.
\label{thm:main1-const}
\eet
\proof Since
\[
\e^{\i sH_\Lambda(\Psi|\Phi,t)}
=\e^{-\i tH_\Lambda(\Phi)}\e^{\i sH_\Lambda(\Psi)}\e^{\i tH_\Lambda(\Phi)},
\]
one has
\beq
\Llim\e^{\i sH_\Lambda(\Psi|\Phi,t)}A\e^{-\i sH_\Lambda(\Psi|\Phi,t)}
=\tau_\Phi^{-t}\circ \tau_\Psi^s\circ \tau_\Phi^{t}(A)
\label{eq:Borodin}
\eeq
for all $A\in\fA$, uniformly for $s$ and $t$ in compacts. Our assumption~\eqref{eq:DynCommut} 
yields that the right-hand side of the last identity is just $\tau_\Psi^s(A)$. 
Since $[H_{\Lambda^\prime\setminus\Lambda}(\Psi|\Phi,t),A]=0$ whenever $A\in\fA_\Lambda$, 
the \preSD{} property further gives that, for such $A$ and all $t$ in a neighborhood of $0$,
\[
\lim_{\Lambda^\prime}[ H_{\Lambda^\prime}(\Psi|\Phi,t),A]
=[H_\Lambda(\Psi|\Phi,t)+W_\Lambda(\Psi|\Phi,t),A].
\]
Hence, on the one hand
\begin{align*}
\tau_\Psi^s(A)
&=\lim_{\Lambda^\prime}\e^{\i sH_{\Lambda^\prime}(\Psi|\Phi,t)}A\e^{-\i sH_{\Lambda^\prime}(\Psi|\Phi,t)}\\[4pt]
&=A+\i\,\lim_{\Lambda^\prime}\int_0^s\e^{\i uH_{\Lambda^\prime}(\Psi|\Phi,t)}
[H_{\Lambda^\prime}(\Psi|\Phi,t),A]\e^{-\i uH_{\Lambda^\prime}(\Psi|\Phi,t)}\d u\\[4pt]
&=A+\i\int_0^s\tau_\Psi^u([H_\Lambda(\Psi|\Phi,t)+ W_\Lambda(\Psi|\Phi,t),A])\d u,
\end{align*}
and on the other hand
\[
\tau_\Psi^s(A)=A+\i\int_0^s\tau_\Psi^u([H_\Lambda(\Psi)+W_\Lambda(\Psi),A])\d u,
\]
which allows us to conclude that
\[
[H_\Lambda(\Psi|\Phi,t)+ W_\Lambda(\Psi|\Phi,t),A]=[H_\Lambda(\Psi)+W_\Lambda(\Psi),A].
\]
Taking the partial normalized trace over $\fA_{\Lambda^c}$ yields
\[
[ H_\Lambda(\Psi|\Phi,t)-H_\Lambda(\Psi)+\tr_{\Lambda^c}(W_\Lambda(\Psi|\Phi,t)-W_\Lambda(\Psi)),A]=0, 
\]
and so 
\beq
H_\Lambda(\Psi|\Phi,t)-H_\Lambda(\Psi)+\tr_{\Lambda^c}(W_\Lambda(\Psi|\Phi,t)-W_\Lambda(\Psi))=C_\Lambda(t)\one
\label{eq:sosot}
\eeq
for some $C_\Lambda(t)\in\RR$. Since $\tr(H_\Lambda(\Psi|\Phi,t)-H_\Lambda(\Psi))=0$, we deduce
\[
\tr(W_\Lambda(\Psi|\Phi,t)-W_\Lambda(\Psi))=C_\Lambda(t).
\]
Hence, by Proposition~\ref{prop:BsStuf}(ii) and the \preSD{} property,
\[
\Lim\frac{|C_\Lambda(t)|}{|\Lambda|}=0.
\]
Going back to~\eqref{eq:sosot} and using contractivity of the partial normalized
trace, we obtain, for any $\nu\in\cSI$ and $t$ in a neighborhood of $0$,
\[
\Lim\frac1{|\Lambda|}
\left|\nu(H_\Lambda(\Psi|\Phi,t)-H_\Lambda(\Psi))\right|=0.
\]
In particular, the corresponding centered-cube limit vanishes. Combining this
with Part~\textbf{(iii)} of the \preSD{} property and
Theorem~\ref{thm:Gibbs}(i) gives
\[
\nu\circ\tau_\Phi^{-t}(E(\Psi))
=\Llim\frac1{|\Lambda|}\nu\circ\tau_\Phi^{-t}(H_\Lambda(\Psi))
=\Llim\frac1{|\Lambda|}\nu(H_\Lambda(\Psi|\Phi,t))
=\nu(E(\Psi)).
\]
Thus a neighborhood of $0$ is contained in $\cT_\Phi(E(\Psi))$ and it follows that
$\cT_\Phi(E(\Psi))=\RR$. \qed

The last result has a partial converse. Denote by $\Psi_t=\left(\Psi_t(X)\right)_{X\in\cF}$
the interaction canonically associated\footnote{Recall Relation~\eqref{eq:incexc}.} to the Hamiltonian $H(\Psi|\Phi,t)$, and note that
Relation~\eqref{eq:Borodin} can be written as 
\beq
\tau_{\Psi_t}^s=\tau_\Phi^{-t}\circ \tau_\Psi^s\circ \tau_\Phi^{t}.
\label{eq:Borodine}
\eeq
\bep\label{prop:hc}
Suppose that $(\Phi,\Psi)$ has the \preSD{} property, that $[\Psi]\in\fC(\Phi)$, and that $\Psi_t\in\cBsd$ for $t$ in
a neighborhood of\, $0$. Then for all $s,t\in\RR$, 
\beq
\tau_\Phi^t\circ\tau_\Psi^s=\tau_\Psi^s\circ\tau_\Phi^t.
\label{eq:DynCommute}
\eeq
\eep
\proof By the group property of the dynamics, it suffices to show that~\eqref{eq:DynCommut}
holds for $t$ and $s$ in a neighborhood of $0$. Since $\Psi_t\in\cBsd$ for small
enough $t$, it follows from Theorem~\ref{thm:Gibbs}(i), the \preSD{} property and
$E(\Psi)$ being a constant of the motion for $\tau_\Phi$ that,
for small enough $t$ and any $\nu\in\cSI$,
$$
\nu(E(\Psi_t))=\Llim\frac1{|\Lambda|}\nu(H_\Lambda(\Psi|\Phi,t))
=\nu\circ\tau_\Phi^{-t}(E(\Psi))=\nu(E(\Psi)).
$$
Proposition~\ref{prop:two-rpe}(v) thus gives $E(\Psi_t)\approx E(\Psi)$.
It then follows from Theorem~\ref{thm:phys-eq} that 
$\tau_{\Psi_t}=\tau_\Psi$, and Relation~\eqref{eq:Borodine} yields the statement.
\qed

Unfortunately, we do not know how to verify the condition
$\Psi_t\in\cBsd$ in Proposition~\ref{prop:hc} in non-trivial examples.
 As shown in the next theorem, conservation of
the specific entropy suffices to obtain the converse to
Theorem~\ref{thm:main1-const}. We also give an alternative proof based on
the thermodynamic formalism of the transported Hamiltonians, without
assuming conservation of the specific entropy. This proof requires a strengthening of the \preSD{} property. The
additional requirement, that $\fAloc$ remains a core for the transported
generator, is also of independent structural interest. We formulate
this condition as follows.

First, we observe that $\alpha^s=\tau_\Phi^{-t}\circ\tau_\Psi^s\circ\tau_\Phi^{t}$ defines
a $C^\ast$-dynamics on $\fA$, and we denote by $\delta_t$ its generator. Clearly
$\Dom(\delta_t)=\tau_\Phi^{-t}(\Dom(\delta_\Psi))$, and 
$\delta_t=\tau_\Phi^{-t}\circ\delta_\Psi\circ\tau_\Phi^{t}$. If the pair $(\Phi,\Psi)$ 
has the \preSD{} property then, for $t$ in a neighborhood of $0$, it follows from~\eqref{eq:Borodine}
that
\[
\delta_t(A)
=\lim_{\Lambda'}\sum_{\twol{X\subset\Lambda'}{X\cap\Lambda\neq\emptyset}}\i[\Psi_t(X),A]
=\i [H_\Lambda(\Psi|\Phi,t)+W_\Lambda(\Psi|\Phi,t),A],
\]
for $A\in\fA_\Lambda$. 
\preASS{SDplus}{SD+}
\begin{definition}\label{SDplus}
We say that a pair $(\Phi,\Psi)$ has the \textbf{(SD+)} property if it has the \preSD{} property and, in addition,
for $t$ in a neighborhood of $0$, $\fAloc$ is a core for $\delta_t$. 
\end{definition}

By conjugation, $\tau_\Phi^{-t}(\fAloc)$ is already a core for $\delta_t$.
The additional requirement is that $\fAloc$ itself be a core for $\delta_t$.

\bet\label{thm:main2-const}
Let $\Phi,\Psi\in\cBsd$ and suppose that $[\Psi]\in\fC(\Phi)$. If either
$(\Phi,\Psi)$ has the \preSDplus{} property or $\tau_\Phi$ preserves the
specific entropy, then for all $s,t\in\RR$, 
$$
\tau_\Phi^t \circ \tau_\Psi^s=\tau_\Psi^s\circ \tau_\Phi^t.
$$
\eet
\proof We first assume that $(\Phi,\Psi)$ has the \preSDplus{} property.
By the group property, it suffices to show that $\alpha=\tau_\Psi$ for small
enough $t$. Although it may happen that $\Psi_t\not\in\cBb$,
the pressure and the specific energy associated with the local Hamiltonian family
$H(\Psi|\Phi,t)$ are well-defined. In the formulas below we denote these quantities,
with a slight abuse of notation, by $P(\Psi_t)$ and $\ce_{\Psi_t}$. This is precisely
the local-Hamiltonian formalism used in the companion paper~\cite{Jaksic2026b}. It is immediate from the
definitions~\eqref{eq:PDef} and~\eqref{eq:ha-t} that
\beq 
P(\Psi_t)=P(\Psi),
\label{eq:pres-t}
\eeq 
while Part~\textbf{(iii)} of the \preSD{} property and $E(\Psi)$ being a constant of the motion for $\tau_\Phi$ 
ensure that for any $\omega\in\cSI$ 
\beq
\ce_{\Psi_t}(\omega)=
\Llim\frac1{|\Lambda|}\omega(H_\Lambda(\Psi|\Phi,t))=\Llim\frac1{|\Lambda|}\omega\circ\tau_\Phi^{-t}(H_\Lambda(\Psi))
=\omega\circ\tau_\Phi^{-t}(E(\Psi))=\omega(E(\Psi))=\ce_{\Psi}(\omega).
\label{eq:energy-t} 
\eeq
We now appeal to Araki--Moriya's thermodynamic formalism. We shall use the terminology and
notation of~\cite{Araki2003} and refer the reader to this article for more details.
It follows from the \preSD{} property that $\Psi_t$ is a translation covariant general potential, 
as defined in~\cite[Section~5.5 and~8.1]{Araki2003}. More precisely, Conditions $(\Phi-\mathrm{a})$ 
and $(\Phi-\mathrm{b})$, \ie $\Psi_t(X)^\ast=\Psi_t(X)\in\fA_X$ for $X\in\cF$, with 
$\Psi_t(\emptyset)=0$, are obviously satisfied. From the formula
$$
\Psi_t(X)=\sum_{Y\subset X}(-1)^{|X|-|Y|}H_Y(\Psi|\Phi,t)
$$
one easily infers Condition $(\Phi-\mathrm{f})$, translation covariance.
Condition $(\Phi-\mathrm{c})$, which expresses evenness in the fermionic
setting, is omitted for quantum spin systems. Finally, since for
$\Lambda\subset\Lambda'\in\cF$,
$$
\sum_{\twol{X\subset\Lambda'}{X\cap\Lambda\neq\emptyset}}\Psi_t(X)
=H_{\Lambda'}(\Psi|\Phi,t)-H_{\Lambda'\setminus\Lambda}(\Psi|\Phi,t),
$$
it follows from Part~\textbf{(i)} of the \preSD{} property that
$$
\lim_{\Lambda'}\sum_{\twol{X\subset\Lambda'}{X\cap\Lambda\neq\emptyset}}\Psi_t(X)
=H_\Lambda(\Psi|\Phi,t)+W_\Lambda(\Psi|\Phi,t),
$$
which yields Condition $(\Phi-\mathrm{e})$. Part~(ii) of the \preSD{} property
provides the surface-energy assumption of~\cite[Proposition~14.1]{Araki2003}.
Moreover, invariance of the trace under finite-volume unitary conjugation and
finite-range approximation in $\cBb$ give
$$
c=\Lim\frac1{|\Lambda|}\tr(H_\Lambda(\Psi|\Phi,t))
=\Lim\frac1{|\Lambda|}\tr(H_\Lambda(\Psi))=\tr(E(\Psi)).
$$
By Part~(ii) of the \preSD{} property, the same trace limit is obtained with
$H_\Lambda(\Psi|\Phi,t)+W_\Lambda(\Psi|\Phi,t)$ in place of
$H_\Lambda(\Psi|\Phi,t)$. Thus, \cite[Proposition~14.1]{Araki2003} applies:
there exists a translation covariant standard potential, \cite[Definition~5.10 and Section~8.1]{Araki2003} 
$\widetilde\Psi_t$ such that, for any $\omega\in\cSI$,
\beq
P(\widetilde\Psi_t)=P(\Psi_t)+c,\qquad \ce_{\widetilde\Psi_t}(\omega)=\ce_{\Psi_t}(\omega)-c.
\label{eq:AMchanges}
\eeq
By~\cite[Section~11 and Proposition~14.1]{Araki2003}, Relations~\eqref{eq:pres-t}, \eqref{eq:energy-t} 
and~\eqref{eq:AMchanges} give that $\widetilde\Psi_t$, $\Psi_t$ and $\Psi$ share the same equilibrium states.
The \preSDplus{} property ensures that for $t$ in a neighborhood of $0$ the $C^\ast$-dynamics $\alpha=\tau_\Phi^{-t}\circ\tau_\Psi\circ\tau_\Phi^{t}$
satisfies properties~(II), (III) and~(IV) of~\cite[Section~1]{Araki2003}. It then follows 
from~\cite[Theorems~A and~B]{Araki2003} that $\tau_\Psi$ and $\alpha$ share the same KMS states. 
Hence, uniqueness of the modular group, Theorem~\ref{thm:israel}~\textbf{(iii)}, yields that
$\alpha=\tau_\Psi$ for small enough $t$.

Suppose now that $\tau_\Phi$ preserves the specific entropy. Choose
$\eta\in\cSeq(\Psi)$.  Conservation of the specific entropy and the assumption
$[\Psi]\in\fC(\Phi)$ give, for all $t\in\RR$,
\[
s(\eta_t)-\eta_t(E(\Psi))
=s(\eta)-\eta(E(\Psi))=P(\Psi).
\]
Thus $\eta_t\in\cSeq(\Psi)$, and Theorem~\ref{thm:KMSeq} implies that
$\eta_t$ is a $1$-KMS state for $\tau_\Psi$. By conjugation, $\eta_t$ is also a $1$-KMS state for
$\tau_\Phi^{-t}\circ\tau_\Psi\circ\tau_\Phi^t$. It follows from Theorem~\ref{thm:israel}~\textbf{(iii)} that
\[
\tau_\Phi^{-t}\circ\tau_\Psi^s\circ\tau_\Phi^t=\tau_\Psi^s
\]
for all $s,t\in\RR$. \qed

We formulate the following conjectures.

\conjecture{Conjecture-(SD)}{conjSD}{Any pair $(\Phi, \Psi)\in\cBsd\times\cBsd$ 
has the \preSD{} property.}

\conjecture{Conjecture-(SD+)}{conjSDplus}{Any pair $(\Phi, \Psi)\in\cBsd\times\cBsd$
has the \preSDplus{} property.}

The strongest result currently available is established in~\cite{Jaksic2026b}.
\bet\label{thm:jpst}
Suppose that $\Phi\in\cBdiam{r}$ for some $r>2d$ and that $\Psi\in\cBsd$.
Then the following hold.
\benr
\item For every $t\in\RR$, $\fAloc$ is a core for the generator of
$\tau_\Phi^{-t}\circ\tau_\Psi\circ\tau_\Phi^t$.
\item If, in addition, $\Psi\in\cBdiam{s}$ for some $s>d$, then the pair
$(\Phi,\Psi)$ has the \preSDplus{} property. In this case, the properties
entering its definition hold for every $t\in\RR$.
\eenr 
\eet
The surface-energy estimates and the core property are established in the
companion paper using Lieb--Robinson bounds~\cite{Nachtergaele2019}. A somewhat
more general surface-energy statement is proved there; here we use only the
form stated above. 

Finally, we observe the following symmetry relation.

\begin{proposition}
Suppose that $\Phi,\Psi\in\cBsd$. Then
\beq
[\Psi]\in\fC(\Phi)\quad\Leftrightarrow\quad [\Phi]\in\fC(\Psi).
\label{eq:t-nice}
\eeq
\end{proposition}

\proof We use a bilinear version of the argument in the proof
of~\cite[Theorem~2.8(2)]{Jaksic2023a}.

We first note that
\[
\sum_{X\ni0}\frac1{|X|}
\sum_{Y:\,Y\cap X\neq\emptyset}
\|[\Phi(Y),\Psi(X)]\|
\leq 2\|\Phi\|_{\mathrm s}\|\Psi\|_{\mathrm s},
\]
since $\sum_{Y:\,Y\cap X\neq\emptyset}\|\Phi(Y)\|
\leq |X|\|\Phi\|_{\mathrm s}$. The finite partial sums of
$E(\Psi)=\sum_{X\ni0}\Psi(X)/|X|$ belong to $\fAloc$ and converge
in norm. The estimate shows that their images under $\delta_\Phi$
also converge in norm. By closedness of $\delta_\Phi$, it follows that
$E(\Psi)\in\Dom(\delta_\Phi)$ and that $\delta_\Phi(E(\Psi))$ is given
by the corresponding absolutely convergent commutator series.
Similarly, $E(\Phi)\in\Dom(\delta_\Psi)$.

Fix $\nu\in\cSI$. Absolute convergence and translation invariance give
\begin{align*}
\nu(\delta_\Phi(E(\Psi)))
&=\sum_{X\ni0}\sum_{Y:\,Y\cap X\neq\emptyset}
\frac1{|X|\,|Y|}\sum_{y\in Y}
\nu\bigl(\i[\Phi(Y-y),\Psi(X-y)]\bigr)\\
&=\sum_{Y\ni0}\frac1{|Y|}
\sum_{X:\,X\cap Y\neq\emptyset}
\nu\bigl(\i[\Phi(Y),\Psi(X)]\bigr)\\
&=-\nu(\delta_\Psi(E(\Phi))).
\end{align*}
For the second equality, set $X'=X-y$ and $Y'=Y-y$. Then
$0\in Y'$ and $-y\in X'$. For each fixed pair $(X',Y')$, there
are exactly $|X'|$ possible values of $y$, cancelling the factor
$1/|X'|$.

Suppose now that $[\Psi]\in\fC(\Phi)$. Differentiating the conservation
identity at $t=0$ gives $\nu(\delta_\Phi(E(\Psi)))=0$ for every
$\nu\in\cSI$. The preceding identity therefore yields
$\nu(\delta_\Psi(E(\Phi)))=0$ for every $\nu\in\cSI$.
Since $\nu\circ\tau_\Psi^t\in\cSI$, we obtain
\[
\frac{\d}{\d t}\nu\circ\tau_\Psi^t(E(\Phi))
=(\nu\circ\tau_\Psi^t)(\delta_\Psi(E(\Phi)))=0.
\]
Thus $[\Phi]\in\fC(\Psi)$. The reverse implication follows by
interchanging $\Phi$ and $\Psi$. \qed

\subsection{Specific relative entropy and constants of the motion}
\label{sec:SE}

\bep\label{prop:wd-pres}
Suppose that $\tau_\Phi$ preserves the specific entropy, and let
$(\rho,\Psi)\in\cSI\times\cBb$ be a regular pair. Then the following
statements are equivalent.
\benr
\item $s(\nu \circ \tau_\Phi^t|\rho)=s(\nu|\rho)$ for all $\nu\in\cSI$
and $t\in\RR$.
\item $[\Psi]\in\fC(\Phi)$.
\eenr
\eep
\proof
The pair $(\rho,\Psi)$ being regular, one has
$s(\nu|\rho)=-s(\nu)+\ce_\Psi(\nu)+P(\Psi)$ for all $\nu\in\cSI$.
Since $\tau_\Phi$ preserves the specific entropy, we further have
\begin{align*}
s(\nu|\rho)
&=-s(\nu\circ\tau_\Phi^t)+\ce_\Psi(\nu\circ\tau_\Phi^t)+P(\Psi)
+\nu\left(E(\Psi)-\tau_\Phi^t(E(\Psi))\right)
\\[4pt]
&=s(\nu\circ\tau_\Phi^t|\rho)
+\nu\left(E(\Psi)-\tau_\Phi^t(E(\Psi))\right).
\end{align*}
Thus, \textbf{(i)} holds iff $E(\Psi)\approx\tau_\Phi^t(E(\Psi))$
for all $t\in\RR$. By the definition of $\fC(\Phi)$, this is precisely
\textbf{(ii)}. \qed

The interest of Proposition~\ref{prop:wd-pres} stems from Theorems~\ref{thm:pres-law} 
and~\ref{thm:pres-law-1}, see Theorem~\ref{thm:main-cm-1} below. 

\conjecture{Conjecture-(R+SE)}{conjRplusSE}{For any $\Phi,\Psi\in\cBsd$ and
$\rho\in\cSeq(\Psi)$, the pair $(\rho,\Psi)$ is regular and
\beq
s(\nu\circ\tau_\Phi^t|\rho\circ\tau_\Phi^t)=s(\nu|\rho),
\label{eq:RplusSE}
\eeq
for all $\nu\in\cSI$ and $t\in\RR$.}

Thus \conjRplusSE{} contains \conjR{} and supplements it with a 
finite-time conservation law for the specific relative entropy along pairs 
of state trajectories.

The strongest result currently available is established in the companion
paper~\cite{Jaksic2026b}.
\bet\label{thm:main-cm-1}
Suppose that $\Phi\in\cBf$ and let $\rho\in\cSeq(\Psi)$.
\benr
\item If $\|\Psi\|_r<r$ for some $r>0$, then $(\rho,\Psi)$ is a regular pair
and there exists $T>0$ such that
\[
s(\nu \circ \tau_\Phi^t|\rho\circ\tau_\Phi^t)=s(\nu|\rho)
\]
for all $\nu\in\cSI$ and $|t|<T$. If, in addition, $\rho$ is
$\tau_\Phi$-invariant, one can take $T=\infty$.
\item If $d=1$ and $\Psi\in\cBf$, then $(\rho,\Psi)$ is a regular pair and
\[
s(\nu \circ \tau_\Phi^t|\rho\circ\tau_\Phi^t)=s(\nu|\rho)
\]
for all $\nu\in\cSI$ and $t\in\RR$.
\eenr
\eet
\proof In Part~\textbf{(i)}, Theorem~\ref{thm:first-s}~\textbf{(iv)} gives
$\rho\in\cSwg(\Psi)$, and hence $(\rho,\Psi)$ is regular. The assertion is
then Theorem~\ref{thm:pres-law}. Part~\textbf{(ii)} follows from
Theorems~\ref{thm:pre-ara} and~\ref{thm:pres-law-1}. \qed

\subsection{Dynamical steady states and constants of the motion}
\label{sec:ess-cm}

\bet\label{thm:main-cm}
Let $\omega_+\in\Sigma_+(\omega,\Phi)\cap\cSeq(\Psi_+)$ for some $\Psi_+\in\cBsd$.
If the pair $(\Phi,\Psi_+)$ has the \preSD{} property, then $[\Psi_+]\in\fC(\Phi)$.
\eet
\proof By Theorem~\ref{thm:KMSeq} , $\omega_+\in\cSKMS(\Psi_+)$. 
Since $\omega_+$ is $\tau_\Phi$-invariant, it is 
also, for any fixed $s$, a $1$-KMS state for the $C^\ast$-dynamics 
$\tau_\Phi^s\circ\tau_{\Psi_+}\circ\tau_\Phi^{-s}$. It follows from 
Theorem~\ref{thm:israel}~\textbf{(iii)} that
\[
\tau_\Phi^s\circ\tau_{\Psi_+}^t\circ\tau_\Phi^{-s}=\tau_{\Psi_+}^t
\]
for all $s,t\in\RR$, and Theorem~\ref{thm:main1-const} yields the statement.\qed

By Theorem~\ref{thm:jpst}~\textbf{(ii)}, the assumption of Theorem~\ref{thm:main-cm} is satisfied 
if $(\Phi,\Psi_+)\in\cBdiam{r}\times\cBdiam{s}$ for some $r>2d$ and $s>d$. 
On the other hand, Proposition~\ref{prop:wd-pres} and Theorem~\ref{thm:main-cm-1} yield

\bet\label{thm-jpst2}
Suppose that $\Phi\in\cBf$ and let $\omega_+\in\Sigma_+(\omega,\Phi)\cap\cSeq(\Psi_+)$. 
\benr
\item If $\|\Psi_+\|_r<r$ for some $r>0$, then $[\Psi_+]\in\fC(\Phi)$.
\item The same conclusion holds if $d=1$ and $\Psi_+\in\cBf$. 
\eenr
\eet
\proof Since $\omega_+$ is $\tau_\Phi$-invariant,
Theorem~\ref{thm:main-cm-1}~\textbf{(i)}, applied with $\rho=\omega_+$, holds for all
$t\in\RR$. Proposition~\ref{prop:wd-pres} then implies that
$[\Psi_+]\in\fC(\Phi)$. Part~\textbf{(ii)} follows in the same way from
Theorem~\ref{thm:main-cm-1}~\textbf{(ii)}, applied with $\rho=\omega_+$. \qed

\subsection{Dynamical Gibbs variational principles}
\label{sec:varcs}

The role of constants of the motion becomes particularly transparent in a
variational characterization of dynamical steady states. We begin with a simple
interaction-level statement, which requires no regularity assumptions and will
be referred to as the potential form of the dynamical Gibbs variational
principle. The state formulation developed below is conceptually deeper and
constitutes the main form of the dynamical Gibbs variational principle used in
this work.

For $\Psi\in\cBb$ such that $[\Psi]\in\fC(\Phi)$, set
\[
s_\mathrm{pot}(\omega|\Psi)\coloneq-s(\omega)+\ce_\Psi(\omega)+P(\Psi),
\]
and
\[
\mathrm{ep_{pot}}(\omega)
\coloneq
\inf\left\{s_\mathrm{pot}(\omega|\Psi)\mid
\Psi\in\cBb,[\Psi]\in\fC(\Phi)\right\}.
\]
The quantity $s_\mathrm{pot}(\omega|\Psi)$ depends only on the
physical-equivalence class of $\Psi$.

\bep\label{prop:DGVP-pot}
For any $\omega_+\in\Sigma_+(\omega,\Phi)$,
\[
s(\omega_+)\le s(\omega)+{\rm ep}_{\rm pot}(\omega).
\]
If $\omega_+\in\cSeq(\Psi_+)$ for some $\Psi_+\in\cBb$ with
$[\Psi_+]\in\fC(\Phi)$, then equality holds and $\Psi_+$ is a minimizer.
Under this assumption, $\Psi\in\cBb$ with $[\Psi]\in\fC(\Phi)$ is a minimizer
if and only if $\omega_+\in\cSeq(\Psi)$. If, in addition,
$\Psi,\Psi_+\in\cBsd$, this is equivalent to $\Psi\sim\Psi_+$.
\eep
\proof
For any $[\Psi]\in\fC(\Phi)$, conservation of $E(\Psi)$ and the Gibbs variational principle give
\[
s_{\rm pot}(\omega|\Psi)
=-s(\omega)+\ce_\Psi(\omega_+)+P(\Psi)
\ge s(\omega_+)-s(\omega).
\]
Taking the infimum over $\Psi$ gives the first assertion. If
$\omega_+\in\cSeq(\Psi_+)$, the Gibbs variational principle is saturated
at $\Psi_+$, and hence equality holds. Under this assumption, a
$\tau_\Phi$-invariant interaction $\Psi$ is a minimizer if and only if the Gibbs variational principle
is saturated at $\omega_+$, i.e., if and only if
$\omega_+\in\cSeq(\Psi)$. The last assertion follows from
Theorem~\ref{thm:phys-eq}. \qed

The elementary Gibbs-variational estimate underlying the last proposition is
the same as the one used in the proof of Theorem~\ref{thm:jpt-o}; here the
additional conservation of $E(\Psi)$ turns it into a variational principle
over $\tau_\Phi$-invariant interactions.

We shall refer to Proposition~\ref{prop:DGVP-pot} as the potential form of the
dynamical Gibbs variational principle. It is useful because it requires no
regularity assumption. 

We now pass to the state formulation, which involves
specific relative entropy and gives the main dynamical Gibbs variational
principle of this work. Define
\[
\fC_\req(\Phi)\coloneq
\left\{[\Psi]\in\fC(\Phi)\,\middle|\, 
(\nu,\Psi)\text{ is regular for every }\nu\in\cSeq(\Psi)\right\}.
\]
Proposition~\ref{prop:two-rpe} shows that both
$\cSeq(\Psi)$ and the involved regularity property are independent of the choice of
the representative $\Psi$. Further set
\[
\cSreq(\Phi)
\coloneq
\bigcup_{[\Psi]\in\fC_\req(\Phi)}\cSeq(\Psi),
\]
and recall the convention $\inf\emptyset=\infty$.
\bep\label{prop:fr}
For any $\omega_+\in\Sigma_+(\omega,\Phi)$,
\[
\inf_{\nu\in\cSreq(\Phi)}s(\omega|\nu)
\ge s(\omega_+)-s(\omega).
\]
If $\omega_+\in\cSeq(\Psi_+)$ for some $\Psi_+$ such that
$[\Psi_+]\in\fC_\req(\Phi)$, then equality holds and every
$\nu\in\cSeq(\Psi_+)$ is a minimizer; in particular $\omega_+$ is a
minimizer. More generally, for $[\Psi]\in\fC_\req(\Phi)$ and
$\nu\in\cSeq(\Psi)$, the state $\nu$ attains the lower bound iff
$\omega_+\in\cSeq(\Psi)$.
\eep
\proof Let $[\Psi]\in\fC_\req(\Phi)$ and $\nu\in\cSeq(\Psi)$. The regularity of
$(\nu,\Psi)$ gives that
\[
s(\omega|\nu)=-s(\omega)+\ce_{\Psi}(\omega)+P(\Psi).
\]
Since $\Psi$ is $\tau_\Phi$-invariant, we further have
$\ce_{\Psi}(\omega)=\ce_{\Psi}(\bar\omega_T)$ for all $T>0$. Taking a
subsequence along which $\bar\omega_T\to\omega_+$ gives
\begin{align*}
s(\omega|\nu)
=-s(\omega)+\ce_{\Psi}(\bar\omega_T)+P(\Psi)
\to&-s(\omega)+\ce_{\Psi}(\omega_+)+P(\Psi)\nonumber\\
&=s(\omega_+|\nu)+s(\omega_+)-s(\omega)
\ge s(\omega_+)-s(\omega).
\end{align*}
For the fixed pair $(\nu,\Psi)$, equality holds iff
$s(\omega_+|\nu)=0$, equivalently iff
$\omega_+\in\cSeq(\Psi)$. The remaining assertions follow immediately.
\qed

Set
\[
\cS_{\req,+}(\Phi)\coloneq\cSreq(\Phi)\cap\Sigma_+(\omega,\Phi).
\]
\bep\label{prop:stein}
Assume that $\cS_{\req,+}(\Phi)\neq\emptyset$.
\benr
\item For any $\omega_+,\omega_+'\in\cS_{\req,+}(\Phi)$,
$$
s(\omega_+)=s(\omega_+'),\qquad s(\omega_+'|\omega_+)=0.
$$
\item Fix $\omega_+\in\cS_{\req,+}(\Phi)$ and $\Psi_+\in\cBb$ such that
$[\Psi_+]\in\fC(\Phi)$, $\omega_+\in\cSeq(\Psi_+)$, and
$(\nu,\Psi_+)$ is regular for every $\nu\in\cSeq(\Psi_+)$. Then
\[
\cS_{\req,+}(\Phi)\subset\cSeq(\Psi_+).
\]
\eenr
\eep
\proof For Part~\textbf{(i)}, choose a representative $\Psi_+$ of a class in
$\fC_\req(\Phi)$ such that $\omega_+\in\cSeq(\Psi_+)$. Proposition~\ref{prop:fr}
gives
$s(\omega_+)=s(\omega_+')$. Since $[\Psi_+]\in\fC(\Phi)$,
\[
\ce_{\Psi_+}(\omega_+')=\ce_{\Psi_+}(\omega)=\ce_{\Psi_+}(\omega_+).
\]
Using regularity of $(\omega_+,\Psi_+)$ we therefore obtain
\[
s(\omega_+'|\omega_+)
=-s(\omega_+')+\ce_{\Psi_+}(\omega_+')+P(\Psi_+)
=-s(\omega_+)+\ce_{\Psi_+}(\omega_+)+P(\Psi_+)=0.
\]
This proves Part~\textbf{(i)}. 

For Part~\textbf{(ii)}, let $\omega_+'\in\cS_{\req,+}(\Phi)$. By
Part~\textbf{(i)}, $s(\omega_+'|\omega_+)=0$. 
Since $(\omega_+,\Psi_+)$ is regular, this is equivalent to
$\omega_+'\in\cSeq(\Psi_+)$. 
\qed

Invoking Proposition~\ref{prop:two-rpe}, we define the physical sector of the 
state space as
\[
\cSphys(\Phi)\coloneq\bigcup_{[\Psi]\in\fCphys(\Phi)}\cSeq(\Psi),
\]
where $\fCphys(\Phi)$ was defined in~\eqref{eq:Cphys-def}.
The following is the state form of the dynamical Gibbs variational principle.
\bep\label{prop:DGVP-state}
Let $\omega_+\in\Sigma_+(\omega,\Phi)\cap\cSphys(\Phi)$ and choose
$\Psi_+\in\cBsd\cap\cBR$ such that
$[\Psi_+]\in\fC(\Phi)$ and $\omega_+\in\cSeq(\Psi_+)$. Then, the following hold.
\benr
\item
\[
\inf_{\nu\in\cSphys(\Phi)}s(\omega|\nu)
=s(\omega_+)-s(\omega).
\]
\item The set of minimizers is precisely $\cSeq(\Psi_+)$. In particular, the
minimizer is unique iff $\cSeq(\Psi_+)=\{\omega_+\}$.
\item If $\Sigma_+(\omega,\Phi)\subset\cSphys(\Phi)$ and
$\cSeq(\Psi_+)=\{\omega_+\}$, then
$\Sigma_+(\omega,\Phi)=\{\omega_+\}$.
\eenr
\eep

\proof Since $\cSphys(\Phi)\subset\cSreq(\Phi)$, Proposition~\ref{prop:fr}
gives
\[
s(\omega|\nu)\ge s(\omega_+)-s(\omega)
\]
for $\nu\in\cSphys(\Phi)$. Moreover, $[\Psi_+]\in\fC_\req(\Phi)$, and hence the same
proposition gives
\[
s(\omega|\omega_+)=s(\omega_+)-s(\omega).
\]
Since $\omega_+\in\cSphys(\Phi)$, this proves Part~\textbf{(i)}.

Let $\nu\in\cSphys(\Phi)$ be a minimizer and choose $\Psi\in\cBsd\cap\cBR$ such that $\nu\in\cSeq(\Psi)$.
Equality in Proposition~\ref{prop:fr} gives $\omega_+\in\cSeq(\Psi)$. Since
also $\omega_+\in\cSeq(\Psi_+)$, Theorem~\ref{thm:phys-eq} yields
$\Psi\sim\Psi_+$ and therefore $\nu\in\cSeq(\Psi_+)$. Conversely, every
$\nu\in\cSeq(\Psi_+)$ is a minimizer by Proposition~\ref{prop:fr}. This proves
Part~\textbf{(ii)}.

For Part~\textbf{(iii)}, let $\omega_+'\in\Sigma_+(\omega,\Phi)$. By assumption it belongs
to $\cSphys(\Phi)$, hence to $\cSreq(\Phi)$. Proposition~\ref{prop:stein}(ii)
then gives $\omega_+'\in\cSeq(\Psi_+)=\{\omega_+\}$. \qed

\section{Approach to equilibrium}
\label{sec:AE}

Throughout this section we assume that $\tau_\Phi$ preserves the specific
entropy.

\begin{definition}\label{def:ae}
The pair $(\omega,\Phi)$ has the property of \textsl{approach to equilibrium} if
\benr
\item the limit
\[
\omega_+=\wstarlim_{T\to\infty}\frac1T\int_0^T\omega\circ\tau_\Phi^t\d t
\]
exists;
\item $\omega_+\in\cSphys(\Phi)$.
\eenr
\end{definition}
Part~\textbf{(i)} is a model-dependent problem; our focus here is on the
structural problem posed by Part~\textbf{(ii)}. Whenever Part~\textbf{(ii)}
holds, Proposition~\ref{prop:DGVP-state} yields the dynamical Gibbs
variational principle.

It remains to understand when a DSS belongs to $\cSphys(\Phi)$. This is where
the structural conjectures enter.

\bep\label{prop:s-ha-2}
Suppose that $\omega_+\in\Sigma_+(\omega,\Phi)\cap\cSeq(\Psi_+)$ with
$\Psi_+\in\cBsd$. If either \conjR{} and \conjSD{} both hold, or
\conjRplusSE{} holds, then $[\Psi_+]\in\fC(\Phi)$ and
$\Psi_+\in\cBR$, and consequently $\omega_+\in\cSphys(\Phi)$.
\eep

\proof Under \conjR, $\Psi_+\in\cBR$. Under \conjSD, the
pair $(\Phi,\Psi_+)$ has the \preSD{} property and
Theorem~\ref{thm:main-cm} gives $[\Psi_+]\in\fC(\Phi)$. This proves the first
case.

If \conjRplusSE{} holds, then $\Psi_+\in\cBR$ and, since a DSS is
$\tau_\Phi$-invariant, Relation~\eqref{eq:RplusSE} gives, for all
$\nu\in\cSI$ and $t\in\RR$,
\[
s(\nu\circ\tau_\Phi^t|\omega_+)=s(\nu|\omega_+).
\]
Proposition~\ref{prop:wd-pres} yields $[\Psi_+]\in\fC(\Phi)$. \qed

We finish with an answer to the question of Ruelle. 
 
\bep Suppose that $\{\omega\}=\cSeq(\Psi_0)$ for some $\Psi_0\in\cBsd$.
If $\omega_+\in\Sigma_+(\omega,\Phi)\cap\cSphys(\Phi)$, then
$s(\omega_+)>s(\omega)$.
\label{prop-ruelle}
\eep
\proof Suppose that $s(\omega)=s(\omega_+)$ for some
$\omega_+\in\Sigma_+(\omega,\Phi)\cap\cSphys(\Phi)$. 
Since $\omega_+\in\cSeq(\Psi_+)$ with $[\Psi_+]\in\fC(\Phi)$, we have
$\omega(E(\Psi_+))=\omega_+(E(\Psi_+))$, and the relation 
\[
s(\omega_+)=\omega_+(E(\Psi_+))+ P(\Psi_+)
\]
gives 
\[
s(\omega)=\omega(E(\Psi_+))+P(\Psi_+).
\]
Hence $\omega\in\cSeq(\Psi_+)$. It follows that $\cSeq(\Psi_0)\cap\cSeq(\Psi_+)\neq\emptyset$, 
and so $\cSeq(\Psi_0)=\cSeq(\Psi_+)$ by Theorem~\ref{thm:phys-eq}. Hence $\omega=\omega_+$ 
which contradicts our basic hypothesis that $\omega$ should not be $\tau_\Phi$-invariant.
\qed
 
 We now give the formal definition of finite-range restoration of locality.

\begin{definition}\label{def:fr-rol}
The pair $(\omega,\Phi)$ has the property of \textsl{finite-range restoration
of locality} if every $\omega_+\in\Sigma_+(\omega,\Phi)$ belongs to
$\cSeq(\Psi_+)$ for some $\Psi_+\in\cBf$.
\end{definition}

In dimension one, the finite-range Gibbs theory leads to the following result.
\bet\label{thm:fr-rol-1d}
Suppose that $d=1$, that $\Phi,\Psi_0\in\cBf$, that
$\omega\in\cSeq(\Psi_0)$, and that $(\omega,\Phi)$ has the property of
finite-range restoration of locality. Then the following hold:
\benr
\item $\Sigma_+(\omega,\Phi)=\{\omega_+\}$.
\item for any $\Psi_+\in\cBf$ such that $\omega_+\in\cSeq(\Psi_+)$, one 
has $[\Psi_+]\in\fC(\Phi)$ and $\Psi_+\in\cBR$. 
\item the dynamical Gibbs variational principle holds and $\omega_+$ is its
unique minimizer.
\item $s(\omega_+)>s(\omega)$.
\eenr
\eet
\proof Let $\nu\in\Sigma_+(\omega,\Phi)$. By finite-range restoration of
locality, $\nu\in\cSeq(\Psi_\nu)$ for some $\Psi_\nu\in\cBf$. By
Theorem~\ref{thm:jpst}~\textbf{(ii)}, the pair $(\Phi,\Psi_\nu)$ has the \preSD{}
property, and it follows from Theorem~\ref{thm:main-cm} 
that $[\Psi_\nu]\in\fC(\Phi)$. By Theorem~\ref{thm:pre-ara}, $\cSeq(\Psi_\nu)$ 
is a singleton and its unique element is weak Gibbs, hence regularity holds. 
Thus $[\Psi_\nu]\in\fCphys(\Phi)$ and $\Sigma_+(\omega,\Phi)\subset\cSphys(\Phi)$.

Fix $\omega_+\in\Sigma_+(\omega,\Phi)$ and a corresponding
$\Psi_+\in\cBf$. Theorem~\ref{thm:pre-ara} gives
$\cSeq(\Psi_+)=\{\omega_+\}$, while Proposition~\ref{prop:DGVP-state}~(iii) yields
Part~\textbf{(i)}. Parts~\textbf{(ii)} and~\textbf{(iii)} follow from the preceding 
argument and Proposition~\ref{prop:DGVP-state}. Finally,
$\cSeq(\Psi_0)=\{\omega\}$ by Theorem~\ref{thm:pre-ara}, and Part~\textbf{(iv)} follows
from Proposition~\ref{prop-ruelle}. \qed

\newpage 

\section{Approach to thermal equilibrium}
\label{sec:ate}\label{sec:ate-revisited}

\begin{definition}\label{def:ate}
The pair $(\omega,\Phi)$ has the property of \textsl{approach to thermal
equilibrium} if
\[
\Sigma_+(\omega,\Phi)\subset\cSeq(\omega,\Phi).
\]
\end{definition}
Unlike Definition~\ref{def:ae}, this definition does not require uniqueness
of the DSS or relative-entropy regularity.

We first prove the finite-range criterion stated in the Introduction.

\bet\label{thm:fr-ate}
Suppose that $\Phi\in\cBf$ and that $(\omega,\Phi)$ has the property of
finite-range restoration of locality. If
\[
\fC_{\cBf}(\Phi)\coloneq\fC(\Phi)\cap[\cBf]=\operatorname{span}\{[\Phi]\},
\]
then $(\omega,\Phi)$ has the property of approach to thermal equilibrium.
\eet
\proof Let $\omega_+\in\Sigma_+(\omega,\Phi)$ and choose
$\Psi_+\in\cBf$ such that $\omega_+\in\cSeq(\Psi_+)$. By
Theorem~\ref{thm:jpst}~\textbf{(ii)} and Theorem~\ref{thm:main-cm}, 
$[\Psi_+]\in\fC(\Phi)$. Hence $[\Psi_+]=\lambda[\Phi]$ for some
$\lambda\in\RR$. Since $\ce_\Phi(\omega_+)=\ce_\Phi(\omega)$, the 
definition of the equilibrium inverse temperature gives $\lambda=\beta_\eq$. 
Thus $\omega_+\in\cSeq(\omega,\Phi)$. \qed

If additional constants of the motion are present, they impose further
constraints on thermalization. Indeed, for any $\nu\in\Sigma_+(\omega,\Phi)$,
$$
\omega(C)=\nu(C)
$$
for all $C\in\fC(\Phi)$.
\bet\label{thm:main-ate}
Suppose that $\beta_\eq \Phi\in\cBR$ and that there exists
$\omega_\eq\in\cSeq(\omega,\Phi)$ such that
\[
\omega(C)=\omega_\eq(C)
\]
provided $[C]\in\fCphys(\Phi)$. Let $\omega_+\in\Sigma_+(\omega,\Phi)$. Then the following statements are
equivalent:
\benr
\item $\omega_+\in\cSeq(\omega,\Phi)$;
\item $\omega_+\in\cSphys(\Phi)$.
\eenr
\eet

\proof
Suppose first that $\omega_+\in\cSeq(\omega,\Phi)$. Then
$\omega_+\in\cSeq(\beta_\eq\Phi)$. Since $E(\Phi)$ is a constant of the
motion and $\beta_\eq\Phi\in\cBsd\cap\cBR$, it follows that
$\omega_+\in\cSphys(\Phi)$.

Conversely, assuming that $\omega_+\in\cSphys(\Phi)$, one can choose 
$\Psi_+\in\cBsd\cap\cBR$ such that $[\Psi_+]\in\fC(\Phi)$ and
$\omega_+\in\cSeq(\Psi_+)$. By the assumption on $\omega_\eq$ and
conservation of $E(\Psi_+)$,
\[
\omega_\eq(E(\Psi_+))
=\omega(E(\Psi_+))
=\omega_+(E(\Psi_+)).
\]
Hence
\begin{align*}
0\le s(\omega_\eq|\omega_+)
&=-s(\omega_\eq)+\omega_\eq(E(\Psi_+))+P(\Psi_+)\\[1mm]
&=s(\omega_+)-s(\omega_\eq).
\end{align*}
Proposition~\ref{prop:st-pres} gives the reverse inequality, and therefore
$s(\omega_\eq|\omega_+)=0$. Regularity of $(\omega_+,\Psi_+)$ yields
$\omega_\eq\in\cSeq(\Psi_+)$. Since also
$\omega_\eq\in\cSeq(\beta_\eq\Phi)$ and
$\Psi_+,\beta_\eq\Phi\in\cBsd$, Theorem~\ref{thm:phys-eq} gives $\Psi_+\sim\beta_\eq\Phi$. 
It follows that $\omega_+\in\cSeq(\beta_\eq\Phi)$. Finally,
$\ce_\Phi(\omega_+)=\ce_\Phi(\omega)$ by energy conservation, and hence
$\omega_+\in\cSeq(\omega,\Phi)$.\vskip-8pt \qed

\printbibliography

\end{document}